\documentclass[lettersize,journal]{IEEEtran}
\usepackage{amsmath,amsfonts, amssymb, amsthm, mathtools}
\usepackage{xcolor}
\usepackage[ruled,vlined]{algorithm2e}
\usepackage{mathrsfs}
\usepackage{array}
\usepackage[caption=false,font=normalsize,labelfont=sf,textfont=sf]{subfig}
\usepackage{textcomp}
\usepackage{stfloats}
\usepackage{url}
\usepackage{placeins}
\usepackage{verbatim}
\usepackage{graphicx}
\usepackage{cite}
\newcommand{\adj}{\a}
\renewcommand{\deg}{\boldsymbol{D}}
\newcommand{\lap}{\boldsymbol{L}}

\renewcommand{\u}{\boldsymbol{u}}
\newcommand{\E}{\boldsymbol{E}}

\renewcommand{\v}{\boldsymbol{v}}
\newcommand{\err}[1]{\mathrm{err}(#1, S)}
\newcommand{\norm}[1]{\left\lVert#1\right\rVert}
\newcommand{\ones}[1]{\boldsymbol{1}_{#1}}
\newcommand{\F}{\mathcal{F}}
\newcommand{\ident}[1]{\boldsymbol{I}_{#1}}
\newcommand{\allones}[1]{\boldsymbol{J}_{#1}}
\newcommand{\allzeros}[1]{\boldsymbol{0}_{#1}}
\newcommand{\vopt}{\boldsymbol{v}_{\mathrm{opt}}}
\newcommand{\GSG}{\boldsymbol{E}^{\mathrm{SG}}}
\newcommand{\GEP}{\boldsymbol{E}^{\mathrm{EP}}}
\newcommand{\GPB}{\boldsymbol{E}^{\mathrm{soft}}}
\newcommand{\GBC}{\boldsymbol{E}^{\mathrm{BGC}}}

\newcommand{\GB}{\boldsymbol{E}^{\mathrm{B}}}
\newcommand{\GF}{\boldsymbol{E}^{\mathrm{FRC}}}
\newcommand{\BE}{\boldsymbol{E}}
\newcommand{\gaussian}[2]{\mathcal{N}(#1, #2)}
\newcommand{\muu}{\boldsymbol{\mu}}
\newcommand{\Sigmaa}{\boldsymbol{\Sigma}}
\newcommand{\floor}[1]{\left\lfloor #1 \right\rfloor}
\newcommand{\W}[1]{\mathcal{W}_{#1}}
\newcommand{\D}[1]{\mathcal{D}_{#1}}
\newcommand{\dd}[1]{\boldsymbol{d}_{#1}}
\newcommand{\xx}[1]{\boldsymbol{x}_{#1}}
\newcommand{\ww}[1]{\boldsymbol{w}_{#1}}
\newcommand{\DD}{\boldsymbol{D}}

\newcommand{\RR}{\mathbb{R}}
\newcommand{\NN}{\mathbb{N}}
\newcommand{\supp}[1]{\mathrm{supp}(#1)}

\newcommand{\prob}[1]{\mathbb{P}\left[#1\right]} 
\newcommand{\ex}[1]{\mathbb{E}\left[#1\right]}

\newcommand{\Cpsitwo}{C_{\psi_2}}
\newcommand{\psione}[1]{\|#1\|_{\psi_1} }
\newcommand{\psitwo}[1]{\|#1\|_{\psi_2} }
\newcommand{\sigmamax}{\sigma_{\max}}
\newcommand{\betamax}{\beta_{\max}}

\newcommand{\errc}{\mathrm{err}_{\rho,\F}(\GSG)}

\renewcommand{\a}{\boldsymbol{a}}
\renewcommand{\adj}{\boldsymbol{A}}
\newtheorem{definition}{Definition}
\newtheorem{theorem}{Theorem}
\newtheorem{lemma}{Lemma}
\newtheorem{corollary}{Corollary}

\newtheorem{remark}{Remark}

\begin{document}
\allowdisplaybreaks
\title{Probabilistic Gradient Coding via Structure-Preserving Sparsification}

\author{Yuxin Jiang, Wenqin Zhang and Lele Wang.~\IEEEmembership{}
        % <-this % stops a space

\thanks{Yuxin Jiang and Lele Wang are with Department of Electrical and Computer Engineering, University of British Columbia, Vancouver, BC, Canada, (email:\{yjiang28, lelewang\}@ece.ubc.ca).
Wenqin Zhang is with School of Cyber Science and Technology, Hubei Univeristy, Wuhan, Hubei, China, (email: wenqin\_zhang@hubu.edu.cn).
}
\thanks{A preliminary version of this work was presented at the 2024 IEEE International Symposium on Information Theory (ISIT), Athens, Greece, and accepted to the 2026 ISIT Workshop on Coding for New Applications.}}

% The paper headers
\markboth{
% Journal of \LaTeX\ Class Files,~Vol.~14, No.~8, August~2021
}%
{Shell \MakeLowercase{\textit{et al.}}: A Sample Article Using IEEEtran.cls for IEEE Journals}

\IEEEpubid{
% \TODO{0000--0000/00\$00.00~\copyright~2021 IEEE}
}
% Remember, if you use this you must call \IEEEpubidadjcol in the second
% column for its text to clear the IEEEpubid mark.

\maketitle

\begin{abstract}
Gradient coding is a distributed computing technique aiming to provide robustness against slow or non-responsive computing nodes, known as stragglers, while balancing the computational load for responsive computing nodes. Among existing gradient codes, a construction based on combinatorial designs, called BIBD gradient code, achieves the best trade-off between robustness and computational load in the worst-case adversarial straggler setting. However, the range of system parameters for which BIBD gradient codes exist is limited. In this paper, we overcome these limitations by proposing two new probabilistic gradient codes, termed the \emph{Sparse Gaussian} (SG) gradient code and the \emph{Expansion-Preserving} (EP) gradient code. 
Through probabilistic constructions, the former preserves the combinatorial structure of BIBDs, while the latter preserves key spectral properties. Both codes are based on a common two-step framework: first generating a random matrix and then applying distinct sparsification procedures. The SG gradient code constructs its encoding matrix from a correlated multivariate Gaussian distribution masked by Bernoulli random variables, while the EP gradient code derives its encoding matrix from sparsified expander-like graph structures that preserve key spectral properties. Experimentally, both codes achieve worst-case error performance comparable to that of the BIBD gradient code (when such a code with the same parameters exists). Moreover, they substantially extend the feasible range of system parameters beyond BIBD and soft BIBD gradient codes, offering practical and theoretically grounded solutions for large-scale distributed computing tasks.
\end{abstract}


\section{Introduction}
Due to the ever-increasing computing demands in various applications such as machine learning and cloud computing, it becomes essential for large systems to perform their computations in a distributed manner. In theory, distributed algorithms can substantially speed up computations compared to centralized algorithms. However, these speed-up gains may not be realized in practice due to the presence of underperforming or non-responsive computing nodes, termed as stragglers\cite{stragglers}. Thus, how to distribute computations that are robust against stragglers is a major challenge in the design of distributed computation systems. In pursuit of this objective, numerous strategies have been developed to mitigate the impact of stragglers. These include duplicating tasks across computing nodes~\cite{aktas2017effective} and, where feasible, excluding stragglers from the computation process~\cite{ananthanarayanan2013effective}. The primary challenge in these approaches lies in how to balance communication efficiency against straggler tolerance~\cite{zaharia2010spark}.

To address this challenge, methodologies grounded in coding theory have been explored. \textit{Gradient Coding} (GC) has emerged as an important tool. It involves dividing each training data into partitions and distributing them to certain computing nodes known as workers. These workers compute and send a linear combination of gradients to the master node, aggregating these to recover the gradient sum accurately. The utmost goal of GC is to recover the gradient sum exactly. Several gradient codes have been proposed under the exact recovery criterion, such as \textit{Fractional Repetition} gradient codes. However, a significant drawback of exact recovery is its potentially high computing workload as it increases with the number of stragglers.

Though exact recovery is desired, approximate recovery may be sufficient in many practical applications, such as stochastic gradient descent algorithms \cite{mania2017perturbed}. Instead of exactly recovering the gradient sum, approximate recovery finds an approximation of the gradient sum \cite{approx-gc}. The 
$\ell_2$-norm of the difference between the exact gradient sum and the approximate gradient sum (i.e. squared error) is a common performance metric. There are two common assumptions for the stragglers: (1) the stochastic stragglers, where computing nodes are straggled independently with some probability, and (2) adversarial stragglers, where stragglers can deliberately disrupt any set of $S$ computing nodes. Throughout this paper, we assume the adversarial straggler setting, which is more realistic for practical systems where obtaining a statistical estimation of the stragglers in real time is difficult. Under this assumption, we maximize the squared error over all possible straggling patterns and refer to it as \textit{the worst-case squared error}.

In the literature, many approximate gradient codes with small worst-case squared errors are introduced, for example, \cite{aktas2017effective, expander-GC, approx-gc, bibd, soft-bibd}.
Among these, one research direction in approximate gradient codes involves the utilization of the biadjacency matrix from expander graphs to develop gradient codes\cite{expander-GC}. Expander graphs are favourable for gradient coding due to their strong connectivity, which enhances the robustness against stragglers. It is shown that the upper bound of the squared error is proportional to the second largest eigenvalue of the normalized adjacency matrix for the expander graph~\cite{expander-GC}. This makes Ramanujan graphs a good candidate as they provide small second largest eigenvalues. As a special type of Ramanujan graphs, \textit{Balanced Incomplete Block Designs} (BIBDs) achieves the best error performance with a closed-form expression for the worst-case error. Moreover, the error depends on the straggling pattern only through the number of stragglers, not the specifics of the straggling pattern~\cite{bibd}.

Unfortunately, most existing constructions for expander graphs require extensive computer search \cite{marcus2013interlacing} and the set of system parameters for which expander graphs exist is limited~\cite{expander-limits}. In particular, the allowable parameters of BIBDs are notably limited~\cite{bibd-para}. 
% Consequently, alternative gradient codes have been developed to emulate BIBD gradient codes by generating probabilistic gradient codes, rather than deterministic ones. This is termed soft BIBD gradient codes in~\cite{soft-bibd}. It is important to note that while these designs expand the range of parameters of BIBDs, they are still restricted to binary gradient codes, that is, the encoding matrix has binary entries. Moreover, the construction of the encoding matrix involves sampling from a complicated high-dimensional distribution, which may be difficult to implement.

This motivates the following question: \emph{by allowing a real-valued encoding matrix instead of a binary one, can one construct, in polynomial time, a gradient code that broadens the feasible range of system parameters while retaining error performance comparable to that of BIBD gradient codes?} We answer this question affirmatively by proposing two probabilistic gradient codes: the \emph{Sparse Gaussian gradient code} (SG-GC) and the \emph{Expansion-Preserving gradient code} (EP-GC). Both constructions follow a common two-step principle: a random matrix is first generated and then sparsified via a code-specific procedure, with each step admitting a polynomial-time algorithm. Moreover, both codes support a wide range of system parameters, as illustrated in Fig.~\ref{fig:para-region} and Fig.~\ref{fig:EP-parameter-regime}. Along the construction process, SG-GC inherits key combinatorial properties of BIBD constructions and achieves error performance comparable to BIBD gradient codes with high probability, as established in Theorem~\ref{thm:SG-GC-error-performance}. For EP-GC, the error can be bounded explicitly, as shown in Theorem~\ref{thm:EP-GC-error-performance}. The performance of both codes is further validated experimentally, as shown by Fig.~\ref{fig:gradient-code-err-performance-comparison}.

\section{Problem Formulation}

Consider a computing machine that includes a master node, 
 $N$ workers, and a training set $\mathcal{D} = \{(\xx{i}, y_i)\}_{i=1}^M$,  where $N, M\in \NN^+$, $\xx{i}\in\RR^d$, and $y_i\in\RR$. The master node intends to train a model using gradient descent by distributing the gradient updates amongst the workers. Specifically, we wish to find the $\ww{}\in\RR^d$ that minimizes the empirical risk, i.e., $\min_{\ww{}\in \RR^d}\frac{1}{M}\sum^{M}_{i=1}\ell(\xx{i}, y_i; \ww{})$, where loss function $\ell(\xx{i}, y_i; \ww{})$ measures the distortion of prediction made by $\ww{}$ on $(\xx{i},y_i)$.

The optimization problem can be solved using the gradient descent algorithm. Specifically, the gradient of the loss at the current model $\ww{}^{(t)}$ is computed as $\sum^M_{i=1}\nabla\ell(\xx{i}, y_i; \ww{}^{(t)})$, and the model is updated iteratively based on this gradient. The gradient sum computation will be a bottleneck when dealing with large training data size $M$. To address this challenge, we distribute gradient computations across multiple workers, where each worker computes some subsets of the $M$ gradients whose sum is returned to the master to update the model.
 
 In a distributed  gradient computation setting, the training set $\mathcal{D}$ can be partitioned to $K$ disjoint partial data set 
 $\mathcal{D}_i$ for $i\in[K]$, where $|\mathcal{D}_i| =M/K$.  Define a partial gradient vector corresponding to $\mathcal{D}_i$ as $\boldsymbol{d}_i \triangleq \sum_{(\xx{i},y_i)\in\mathcal{D}_i}\nabla\ell(\xx{i},y_i;\ww{})$. 
 
 In the context of $S$ stragglers, a gradient coding scheme contains two steps: distributing data to workers based on an encoding matrix and approximating the gradient sum from the received $N-S$ partial gradients.

A gradient code can be characterized by an encoding matrix $\BE\in \RR^{K\times N}$, where row $i\in [K]$ corresponds to data partition $\D{i}$ and column $j\in [N]$ corresponds to worker $j$. Worker $j$ will return $\sum_{i=1}^K \BE_{ij} \boldsymbol{d}_i$ where $\BE_{ij}$ is the linear combination coefficient of the partial gradient corresponding to $(\boldsymbol{x}_i, y_i)$.

For a gradient code, each data piece is distributed to an equal number of $R$ workers 
% (the number of non-zero entries in each row is $R$)
and each worker computes an equal number of $L$ partial gradients 
% (the number of non-zero entries in each column is $L$)
. The goal of the master node is to recover the gradient sum exactly
\begin{equation*}
     \sum_{i = 1}^K \dd{i} =\DD\ones{K},
\end{equation*}
where $\DD = [\dd{1}, \dd{2}, \dots, \dd{K}]$. 
 % Let $\F \subset [N]$ be the set of non-straggling workers. We have $|\F| = N-S$. If $\F\not = \emptyset$, $\BE$ is reduced to a sub-matrix $\boldsymbol{E}_{\F} \in\mathbb{R}^{K\times (N-S)}$ with columns indexed by $\F$.
 % Therefore, the master node will receive $\DD \boldsymbol{E}_{\F}$ where each entry is a partial gradient sum by a non-straggling worker. To approximate the gradient sum, the master node finds a linear combination $\v \in \mathbb{R}^{N-S}$ of the received gradient vector $\DD \boldsymbol{E}_{\F} \v$, which is also referred as a decoding vector of $\boldsymbol{E}_{\F}$.
% \lele{[There is a gap in the flow from the previous paragraph to this paragraph.]}
Though an exact recovery of the full gradient sum is ideal, it is often unrealistic in practice. Fortunately, many practical applications only require an approximate recovery. This shifts the goal toward producing an estimate of the gradient sum that minimizes the least-squares error relative to the exact value.

Let $\F \subset [N]$ denote the set of non-straggling workers, with $|\F| = N-S$.
Rather than resizing the encoding matrix $\BE \in \mathbb{R}^{K\times N}$ for each $\F$, 
we keep the full matrix $\BE$ and encode the non-straggler pattern in the decoding vector.
% Specifically, we identify $\mathbb{R}^{|\F|}$ with the subspace $\{x \in \mathbb{R}^N : \supp{x} \subseteq \F\}$ via zero-padding on $\F^c$. For any $\v \in \mathbb{R}^{|\F|}$, let $\v_\F \in \mathbb{R}^N$ denote its zero-padded lift, i.e., $(\v_\F)_i = \v_i$ for $i \in \F$ and $(\v_\F)_i = 0$ for $i \notin \F$.} 
% \lele{[The index for $(\v_\F)_i$ and the index for $\v_i$ are different.]}
% Thus, the difference between the target sum and the approximate sum can be expressed as
% \begin{equation*}
%     \DD \boldsymbol{E}_{\F} \v - \DD\ones{K} = \DD(\boldsymbol{E}_{\F}\v -\ones{K}).
% \end{equation*}
% Since $\DD$ is unknown to us, the goal is to find a decoding vector that minimizes the squared $\ell_2$-norm of the difference. So, given an encoding matrix $\BE\in\RR^{K\times N}$ and non-straggler pattern $\F \subset [N]$, the optimal decoding vector is
%     $$\vopt(\boldsymbol{E}_{\F}) \triangleq \arg \min_{\v \in\mathbb{R}^{|\F|}} \norm{\boldsymbol{E}_{\F}\v - \ones{K}}_2^2.$$
% The performance metric is defined as follows.
Hence, the master node receives the partial gradient sums $\DD \BE \v_\F$, 
where each nonzero coordinate corresponds to a non-straggling worker. 
To approximate the gradient sum, the master node finds a decoding vector $\v_\F \in \mathbb{R}^N$ 
supported on $\F$ that minimizes the squared $\ell_2$-norm of the reconstruction error:
\begin{equation*}
    \vopt(\F) \triangleq \arg\min_{\substack{\v_\F \in \mathbb{R}^N \\ \supp{\v_\F} \subseteq \F}} 
    \|\BE \v_\F - \ones{K}\|_2^2.
\end{equation*}
The difference between the target gradient sum and the approximate reconstruction can be expressed as
\begin{equation*}
    \DD \BE \v_\F - \DD\ones{K} = \DD(\BE\v_\F -\ones{K}).
\end{equation*}
Since $\DD$ is unknown, we evaluate the worst-case decoding error over all possible non-straggler sets.

% \begin{definition}
% \label{def:error}
%     For an encoding matrix  $\BE\in\RR^{K\times N}$, the normalized worst-case squared error (the error) with $S$ stragglers is defined as
% \begin{equation}
%  \err{\boldsymbol{E}} = \frac{1}{K}
% \max_{\substack{\F \subset [N]\\ |\F| = N-S}}
% \norm{\boldsymbol{E}_{\F}\vopt - \ones{K}}^2_2.
% \end{equation}
% \end{definition}
\begin{definition}
\label{def:error}
For an encoding matrix $\BE\in\mathbb{R}^{K\times N}$, 
the normalized worst-case squared error (the \emph{error}) with $S$ stragglers is defined as
\begin{equation}
\err{\BE}
= \frac{1}{K}\max_{\substack{\F \subset [N]\\ |\F| = N-S}}
\|\BE\vopt(\F) - \ones{K}\|_2^2.
\end{equation}
\end{definition}

When the encoding matrix $\BE$ is random, the above error is a random variable. We will study the high probability upper bound on the random error.

Throughout the paper, $\allones{n\times m}$ denotes the $n\times m$ all-one matrix, $\allzeros{n\times m}$ the $n\times m$ all-zero matrix, $\ident{n}$ denotes an $n\times n$ identity matrix and $\ones{n}$ the all-one column vector in $\mathbb{R}^n$. The $\ell_p$-norm is written as $\norm{\cdot}_p$. For $x>0$, $\ln(x)$ denotes the natural logarithm and $\log(x)$ the base-2 logarithm. For a matrix $\boldsymbol{A}\in \RR^{n\times m}$, its singular values are denoted as $\sigma_1(\boldsymbol{A}) \geq \dots \geq \sigma_{\max\{n, m\}}(\boldsymbol{A})$ and its matrix norm is defined as $\norm{\boldsymbol{A}}_2 = \sigma_1(\boldsymbol{A})$. If $\boldsymbol{A}\in \RR^{n\times n}$ is non-singular, its eigenvalues are denoted as $\lambda_1(\boldsymbol{\boldsymbol{A}}) \geq \dots \geq \lambda_n(\boldsymbol{A})$ and thus its matrix norm is defined as $\norm{\boldsymbol{A}}_2 = \lambda_1(\boldsymbol{A})$.

% \jossierevised{Throughout the paper, a graph is denoted as $G$.}
% Let $\BE$ be a matrix, $\BE^\top$ be its transpose and $\BE_{ij}$ be its $(i,j)$-th entry. 
% \jossierevised{
% % If $\BE$ is square and diagonalizable, 
% \lele{Denote by $\lambda_i(\BE)$ the eigenvalues of $\BE$ when they exist,}
% % its eigenvalues are denoted by $\lambda_i(\BE)$, 
% ordered such that $\lambda_1(\BE) \ge \lambda_2(\BE) \ge \cdots$. \lele{[A matrix does not need to be diagonalizable to have eigenvalues.]}
% If $\BE$ is not necessarily invertible, \lele{[Invertibility of a matrix is a different concept. A matrix can be not invertible, but has eigenvalues.]} $\sigma_i(\BE)$ denotes its $i$-th singular value in descending order, and the corresponding left and right singular vectors are written as $\u_i(\BE)$ and $\v_i(\BE)$, satisfying $\BE\v_i(\BE) = \sigma_i(\BE)\u_i(\BE)$ and $\BE^\top\u_i(\BE) = \sigma_i(\BE)\v_i(\BE)$.}
% For positive integers $n$ and $k$, we denote an $n\times n$ identity matrix by $\ident{n}$, an $n\times k$ all-one matrix by $\allones{n\times k}$ and an $n\times k$ all-zero matrix by $\allzeros{n\times k}$. Let $\v \in \RR^{n}$ be a column vector and 
% $\ones{n}$ be an all-one column vector. The $\ell_p$-norm is represented by $\norm{\cdot}_p$. \lele{[This paragraph for notation should go after the problem formulation.]}

\section{Existing Gradient Codes}
This section summarizes the constructions of four existing gradient codes.
\subsection{Fractional Repetition Code}
The \textit{Fractional Repetition Code} (FRC) is the first gradient code proposed in~\cite{frc}. The encoding matrix of an $(N, K, L, R)$-FRC gradient code is given by
\begin{equation*}
\GF = \begin{bmatrix}
        \allones{L\times R} & \allzeros{L\times R} & \cdots & \allzeros{L\times R} \\
        \allzeros{L\times R} & \allones{L\times R} & \cdots & \allzeros{L\times R} \\
        \vdots & \vdots & \ddots & \vdots \\
        \allzeros{L\times R} & \allzeros{L\times R} & \cdots & \allones{L\times R}
        \end{bmatrix},
\end{equation*}

 The empirical analysis demonstrates that the average error of the FRC is commendable. Nevertheless, it is vulnerable to adversarial stragglers, which gives the exact error:
 \begin{equation*}
    \err{\GF} = \frac{L}{K}\floor{\frac{S}{R}}.
\end{equation*}

\subsection{Bernoulli Gradient Code and its regularized variant}
To address the drawback inherent in the FRC, a probabilistic gradient code called the \textit{Bernoulli Gradient Code} (BGC) is proposed in~\cite{BGC}. Each entry in the encoding matrix $\GBC$ of an $(N,K,L)$-BGC is generated i.i.d. according to a Bernoulli distribution $\boldsymbol{E}_{ij}^\mathrm{BGC}\sim$ Bern$(\frac{L}{K})$.

To restrict the computation load per worker to a maximum of $2L$, \textit{regularized Bernoulli Gradient Code} (rBGC) is proposed as a variant of BGC. An $(N, K, L)$-rBGC begins with the same initial structure as an $(N, K, L)$-BGC and then sets entries to zero in columns that exceed $2L$ non-zero entries until only $L$ non-zero entries remain.

\subsection{Bipartite-expander-based Gradient Code\cite{expander-GC}}
The \textit{Bipartite-expander-graph-based Gradient Code} (BEG-GC) leverages the adjacency matrix or biadjacency matrix of biregular bipartite expander graphs. In this paper, we will primarily focus on the one adopting the biadjacency matrix.

An $(N, d)$-BEG-GC corresponds to a $d$-regular bipartite expander graph $G = (\W{}\cup \D{},\mathcal{E})$, where $|\W{}| = |\D{}| = N$. Let $\E^{\mathrm{BEG}}$ be the biadjacency matrix of $G$. The normalized biadjacency matrix $\bar\E^{\mathrm{BEG}} \triangleq \tfrac{1}{d} \E^{\mathrm{BEG}}$ is the encoding matrix for the BEG-GC. Fix a non-straggler set $\F \subset [N]$ and $|\F| = N-S$. Its decoding vector takes the form
\begin{equation}
    \v_\F^{\mathrm{BEG}} = \ones{N}+\a_{\F}
    , \text{ where}\quad
    \a_\F = 
    \begin{cases}
        -1, \quad i\not \in \F\\
        \frac{S}{N-S} \quad i\in \F.
    \end{cases}
\label{eq:decoding_vec}
\end{equation} Its error is upper bounded as
\begin{equation*}
%\label{eq:expander-code-error}
    \err{\bar\E^{\mathrm{BEG}}} \leq \tfrac{1}{N}\Big(\tfrac{\sigma_2(\bar\E^{\mathrm{BEG}})}{d}\Big)^2\tfrac{NS}{N-S}.
\end{equation*}%

\subsection{BIBD Gradient Code}
A $(v, b, k, r, \lambda)$-BIBD is characterized by a set $\mathcal{V}$ of $v$ points and a collection $\mathcal{B}$ of $b$ blocks, each containing exactly $k$ points, such that each point appears in exactly $r$ blocks and every pair of distinct points is included in exactly $\lambda$ blocks~\cite{bibd-design}. The incidence matrix $\boldsymbol{M}$ for this design is a $ v \times b$ binary matrix, where $\boldsymbol{M}_{ij} = 1$ if point $i$ is in block $j$, and $\boldsymbol{M}_{ij} =0$ otherwise. A BIBD is \textit{symmetric} when $\boldsymbol{M} = \boldsymbol{M}^\top$. Given a BIBD with incidence matrix $\boldsymbol{M}$, its dual-design is defined by a design with incidence matrix $\boldsymbol{M}^\top$. Throughout the paper, we only consider symmetric BIBD design or dual BIBD design.

Given the incidence matrix $\boldsymbol{M}$ of a BIBD design, a gradient code $(N = b, K = v, L = k, R = r)$-GC is constructed by setting the encoding matrix $\GB = \boldsymbol{M}$. A sufficient condition for the existence of such code are: for $N, K, L, R, \lambda \in \mathbb{N}^+$, $NL = KR$ and $R(L-1) = \lambda(K-1).$ This code is known to be robust against adversarial stragglers. The corresponding worst-case error when there are $S$ stragglers is given by
\begin{equation}
\label{eq:bibd-error}
\err{\GB} = 1 - \frac{1}{K}\frac{L^2 (N-S)}{L+\lambda(N-S-1)}.
\end{equation}
Moreover, it is known that for any straggling pattern $\F$ with $|\F| = N-S$, the corresponding optimal decoding vector $\vopt(\GB_\F)$ is the same constant vector and the corresponding squared error is identical.

\subsection{Soft BIBD Gradient Code}
A Soft BIBD gradient code arises from the aim to extend the range of system parameters inherent in BIBD gradient codes, while preserving its superior error performance~\cite{soft-bibd}. An $(N, K, L, R)$-Soft BIBD employs a probabilistic approach to generate its encoding matrix, denoted by $\GPB$, which approximates the encoding matrix of an $(N, K, L, R, \lambda)$-BIBD gradient code.

It is shown in~\cite{soft-bibd} that the average error of the soft BIBD gradient code can be as good as that of BIBD with the same parameters. However, the construction of soft BIBD requires the encoding matrix to have binary entries. This potentially limits the set of system parameters the gradient code can take because the encoding matrix can take any real-valued entries in the general gradient code framework. So, in the remaining part of this paper, we demonstrate how the set of system parameters can be enlarged by exploiting real-valued encoding matrices.

\section{Sparse Gaussian Gradient Code}
\label{sec:SG-GC}
% \lele{In this section, we propose the sparse Gaussian gradient code, which ... [We need a transition paragraph here. Please briefly summarize its properties.]}
In this section, we propose the sparse Gaussian gradient code, which extends the parameter regime inherited from the BIBD gradient code while preserves comparably good theoretical performance. BIBD gradient codes offer strong guarantees but exist only for limited parameter regimes, while random constructions are more flexible yet difficult to control. A real-valued framework bridges this divide. By combining prescribed sparsity with Gaussian weights, the sparse Gaussian gradient code broadens the achievable parameter range while preserving analytical tractability and strong approximation performance.

\subsection{Construction}

\label{sec:SG-GC-construction}
The encoding matrix $\GSG$ is constructed by the element-wise product of two $K\times N$ matrices $\boldsymbol{X}$ and $\boldsymbol{B}$, i.e., 
\begin{equation}
    \GSG_{ij} = \boldsymbol{X}_{ij} \boldsymbol{B}_{ij}.
\end{equation}
% where $\odot$ denotes the element-wise multiplication, i.e., $(\boldsymbol{X}~\odot~\boldsymbol{B})_{ij}= ~\boldsymbol{X}_{ij}\boldsymbol{B}_{ij}$.

To generate $\boldsymbol{X}$, we choose a correlated multivariate normal distribution $\gaussian{\muu}{\Sigmaa}$ as the joint distribution for each row, where $\muu \in \RR^{N}$ and $\Sigmaa \in \RR^{N\times N}$ defined as 
\begin{equation}
     \muu = a\ones{N}
\end{equation}
and 
\begin{equation}
\label{def:cov-matrix}
    \Sigmaa = \begin{bmatrix}
        b & c & \cdots & c \\
        c & b & \cdots & c \\
        \vdots & \vdots & \ddots & \vdots \\
        c & c & \cdots & b
        \end{bmatrix} .
\end{equation}
The $K$ row vectors
\[
\left[\boldsymbol{X}_{i1}, \boldsymbol{X}_{i2}, \dots, \boldsymbol{X}_{iN}\right], \quad i \in [K],
\]
are generated i.i.d. according  $\gaussian{\muu}{\Sigmaa}$. Note that different rows of $\boldsymbol{X}$ are independent, while different columns of $\boldsymbol{X}$ are correlated due to the correlation in $\Sigmaa$.

Now, to generate matrix $\boldsymbol{B}$, each entry $~\boldsymbol{B}_{ij}\overset{\mathrm{i.i.d.}}{\sim}\mathrm{Bern}~(\gamma)$. This way, we ensure that a $\gamma$ fraction of the encoding matrix is non-zero in expectation. 

For the construction of an $(N, K, L, R, \lambda, \gamma)$-SG gradient code which simulates an $(N, K, L, R, \lambda)$-BIBD gradient code, three criteria have to be satisfied.
\begin{enumerate}
    \item To mimic the BIBD property that each column has $L$ ones, we require for any $i \in [K]$ and $j \in [N]$,
    \begin{equation}
    \label{row-exp}
        \ex{\GSG_{ij}} = L/K;
    \end{equation}
    \item To mimic the BIBD property that the intersection between any two distinct columns is $\lambda$, we require for any $i\in [K]$ and $j,k \in [N]$, $ j \neq k,$
    \begin{equation}
    \label{diff-intsec-exp}
        \ex{\GSG_{ij}\GSG_{ik}} = \lambda/K;
    \end{equation}
    \item 
    To ensure the intersections between any column and itself is $L$, we require for any $i\in [K]$ and $j \in [N]$,
    \begin{equation}
    \label{same-intsec-exp}
         \ex{\GSG_{ij}\GSG_{ij}} = L/K.
    \end{equation}
\end{enumerate}

Subsequently, for $i\in [K]$ and $ j,k \in [N]  $, $j\not=k$: $~\ex{\boldsymbol{X}_{ij}\boldsymbol{B}_{ij}}=L/K$, $~\ex{\boldsymbol{X}_{ij}\boldsymbol{X}_{ik}\boldsymbol{B}_{ij}\boldsymbol{B}_{ik}}=\lambda/K$ and $\ex{\boldsymbol{X}_{ij}^2\boldsymbol{B}_{ij}^2} = L/K$ can be obtained from Equations~\eqref{row-exp},~\eqref{diff-intsec-exp} and~\eqref{same-intsec-exp}, respectively.
Since the entries of the encoding matrix are independent of the i.i.d. Bernoulli random variables, a system of equations is shown as follows:
\begin{equation}
\label{eq:sys-para-mid}
    \begin{cases}
        a\gamma \overset{(a)}{=} L/K\\
        (c+a^2) \gamma^2 \overset{(b)}{=} \lambda/K\\
        (b+a^2)\gamma \overset{(c)}{=} L/K
    \end{cases},
\end{equation}
where (a), (b) and (c) follow from $\ex{\boldsymbol{X}_{ij}} =a$, $\ex{\boldsymbol{X}_{ij}\boldsymbol{X}_{ik}} =\mathrm{Cov}(\boldsymbol{X}_{ij}\boldsymbol{X}_{ik})+ \ex{\boldsymbol{X}_{ij}}\ex{\boldsymbol{X}_{ik}} =c+a^2$ and $\ex{\boldsymbol{X}_{ij}^2} = \mathrm{Var}(\boldsymbol{X}_{ij}) + \ex{\boldsymbol{X}_{ij}}^2 = b + a^2$, respectively.
Hence, we have
\begin{equation}
\label{eq:sys-para}
    \begin{cases}
        a = L/(K\gamma)\\
        b = (\frac{L\gamma}{K} - \frac{L^2}{K^2})/\gamma^2\\
        c = (\frac{\lambda}{K} - \frac{L^2}{K^2})/\gamma^2
    \end{cases},
\end{equation}
where $\gamma$ is restricted to some range, which will be discussed in the following subsection.

\subsection{Feasible Parameter Regime}
In the following, we provide a sufficient condition for the existence of the above Sparse Gaussian gradient code.
\label{sec:existence}
\begin{theorem}
\label{thm:para-region}
    An $(N, K, L, R, \lambda, \gamma)$ Sparse Gaussian gradient code exists if $L \leq K$ and
    \begin{align*}
        % L&\leq K,\\
        \max\left\{\frac{\lambda}{L},  \frac{(N-1)(L^2-K\lambda) + L^2}{KL}\right\} &\leq \gamma \leq 1.
    \end{align*}
\end{theorem}

Guided by Theorem~\ref{thm:para-region}, Fig. \ref{fig:para-region} provides a visualization of a feasible parameter space for $L$ and $R$ in sparse Gaussian gradient codes with specified values of $N, L/K$ and $R/N$. In comparison, we plot several BIBDs with similar density $L/K$ and the parameter region for soft BIBD gradient codes with the same settings of $N, L/K$ and $R/N$. As illustrated, the parameter region for the proposed Sparse Gaussian gradient code expands the choices of feasible parameters compared to those of the soft BIBD and the BIBD gradient codes. 

\begin{figure}[!tbp]
    \centering
    \includegraphics[width=1\linewidth]{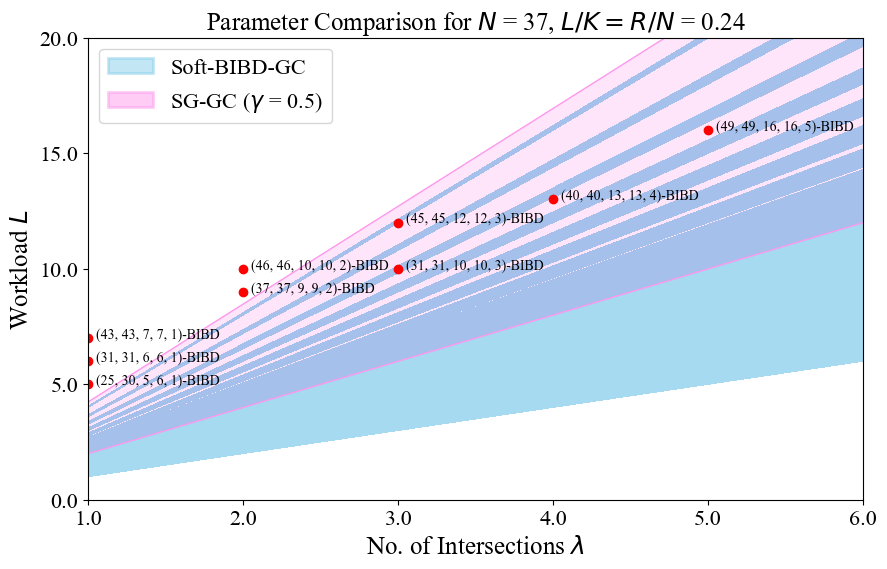}\caption{The region in pink is the parameter region for our proposed SG gradient codes. In contrast, the region in blue is the parameter region for soft-BIBD gradient codes with the probabilistic method. The red dots are combinatorial BIBDs with a similar density, i.e., $0.24\pm 0.05$.}
    \label{fig:para-region}
\end{figure}%

\subsection{Theoretical Performance Evaluation}
\label{sec:error-comparison}
In this section, we provide a high probability upper bound on the \emph{random} worst-case error of the sparse Gaussian gradient code. In the regime of parameters specified in Theorem~\ref{thm:SG-GC-error-performance}, the proposed sparse Gaussian gradient code can have an error performance comparable to that of the symmetric or dual BIBD gradient code with the same parameters.
\begin{theorem}
\label{thm:SG-GC-error-performance}
Consider $(N, K, L, R, \lambda, \gamma)$ Sparse Gaussian gradient codes with encoding matrix $\GSG$ and $K = \Theta(N), L = \Theta(N), \lambda = \Theta(N), \gamma = \Theta(1)$. For $0<\delta<1/2$ and $0<\alpha<2\delta$, let the number of stragglers $S = O(N^{\alpha})$ and $\rho = \frac{L}{L+\lambda(N-S-1)}$. If $\rho^2(N-~S)(b+(N-S-1)c)<2$ and $c \ge 0$, then
\begin{align*}
    &\prob{\err{\GSG}\le \err{\GB}+\tfrac{1}{K^{1/2-\delta}}}
    \ge 1-e^{-\Theta(N^{2\delta})},
\end{align*}
where $\err{\GB}$ is the worst-case error of symmetric or dual $(N, K, L, R, \lambda)$-BIBD gradient code given in Equation~\eqref{eq:bibd-error}.
\end{theorem}
\begin{remark}
    Note that $\frac{1}{K^{1/2-\delta}} =o(1)$ and $e^{-\Theta(N^{2\delta})} = o(1)$. This means in the specified parameter regime, with high probability, the error associated with the sparse Gaussian gradient code concentrates at or below the error of the BIBD Gradient Code with the same parameters.
\end{remark}
The proof of Theorem \ref{thm:SG-GC-error-performance} is provided in Appendix \ref{proof of thm:SG-GC-error-performance}.

\section{Expansion-Preserving Gradient Code Construction}
% \lele{[A transition paragraph here. Briefly explain the limitation of the first construction, and motivate the second construction.]}
In this section, we introduce a new gradient coding scheme called \textit{expansion-preserving gradient code} (EP-GC). 
We notice that the error of an BEG-GC can be upper bounded by a function of the second largest eigenvalue of its encoding matrix when all other parameters are fixed. The EP-GC builds on this insight by first generating an initial encoding matrix with a small second-largest eigenvalue (e.g., via repeated sampling), and then appending a row and column to enforce fixed row and column sums. This structure allows the use of \textsc{DegreePreservingSparsify}, which both sparsifies the matrix and tightly controls changes to the second largest eigenvalue. As a result, we can track and bound the corresponding change in error throughout the sparsification process.

\subsection{Construction}
\label{sec:sparse-graph-gradient-code-encoding-matrix-construction}
Formally, the encoding matrix of an $(N,c, d,\varepsilon)$-EP-GC is constructed in the following three steps:
\begin{enumerate}
    \item \textbf{Initial Matrix Generation:} 
    % Construct a symmetric $(N-1)\times (N-1)$ random matrix
    % \[
    %     \GEP_0 = [|\GEP_{ij}|]_{1\le i,j\le N-1}, \quad 
    %     a_{ij} = a_{ji},
    % \]
    % where $\GEP_{ij} \overset{\text{i.i.d.}}{\sim} \mathcal{N}(0,1) \text{ for } i\le j$
    % and $\GEP_{ij} = \GEP_{ji}$ for $i > j$.
    % This ensures $\GEP_0$ is symmetric with independently sampled half-normal distributed
    % entries in its upper triangular part.
    Construct a symmetric $(N-1)\times(N-1)$ random matrix $\GEP_0$ whose upper-triangular entries are i.i.d. standard half-normal random variables. Specifically, for $1\le i\le j\le N-1$,
\[
(\GEP_0)_{ij} = |X_{ij}|+c,\qquad X_{ij}\overset{\mathrm{i.i.d.}}{\sim}\mathcal{N}(0,1).
\]
% and $(\GEP_0)_{ji}=(\GEP_0)_{ij}$ for $i<j$.
    \item \textbf{Row and Column Sum Adjustment:}
Extend $\GEP_0$ to an $N\times N$ symmetric matrix $\GEP$ by appending a final row and column as
\begin{equation}
\label{eq:extended-sparse-graph-scheduling-matrix}
\GEP =
\begin{bmatrix}
\GEP_0 & \boldsymbol{m} \\
\boldsymbol{m}^\top & \alpha
\end{bmatrix},
\end{equation}
where $\boldsymbol{m}\in\mathbb{R}^{N-1}$ and $\alpha\in\mathbb{R}$ are chosen such that row and column sums equal to $d$.
% \[
% \sum_{j=1}^N \GEP_{ij} = \sum_{i=1}^N \GEP_{ij} = d, \quad \forall i.
% \]
Specifically, let
\begin{equation}
\label{eq:append-col}
    \textstyle
\boldsymbol{m}_i = d - \sum_{j\in[N-1]} (\GEP_0)_{i j},
\quad \forall i
\end{equation}
and
\begin{equation}
\label{eq:alpha}
    \textstyle
\alpha = (2-N)d + \sum_{i, j\in[N-1]} (\GEP_0)_{ij}.
\end{equation}
% Then $\GEP$ is symmetric with all row and column sums equal to $d$.
    % ensures that $\GEP$ is symmetric and has constant row sum and column sum $d$. 
    \item \textbf{Degree-Preserving Sparsification:} 
    Consider the matrix $\GEP$ as the biadjacency matrix of the graph $G$.
    Then, the adjacency matrix of the bipartite graph $G$ is
    \begin{equation*}
    %\label{eq:adj-before-sparsification}
    \adj = \begin{bmatrix}
    \allzeros{} &  \GEP\\
        (\GEP)^\top & \allzeros{}
    \end{bmatrix}.
    \end{equation*}
    Note that the degree matrix of $G$ is $\deg~= ~d\ident{2N}$.
    % where $\ident{2N}$ is a $2N\times 2N$ identity matrix.
    To reduce density while preserving expansion properties, we apply \textsc{DegreePreservingSparsify}~\cite{degree-preserving-sparsifier} to obtain $G_\varepsilon = \textsc{DegreePreservingSparsify}(G,\varepsilon)$. The corresponding biadjacency matrix $\GEP_\varepsilon$ is used as the final encoding matrix.
    \end{enumerate}
     
We note that $G_\varepsilon$ is \emph{degree-preserving $\varepsilon$-sparsifier} of $G$ with $\varepsilon\in(0,1]$ satisfying the following properties:
    \begin{enumerate}
    \item $G_\varepsilon$ is a subgraph of $G$, and thus is also a bipartite graph. For $G_\varepsilon$, we denote its adjacency matrix as $\adj_\varepsilon$ and biadjacency matrix as $\GEP_\varepsilon$.
    \item 
    The weighted degree of each vertex remains unchanged (degree-preserving), i.e. for the degree matrix of $G_\varepsilon$:
    \begin{align}
    \label{eq:degree-preserved}
        \deg_\varepsilon = \deg.
    \end{align}
    \item For the Laplacians of $G$ and $G_\varepsilon$, denoted as $\lap\triangleq \deg - \adj$ and $\lap_{\varepsilon}\triangleq \deg_{\varepsilon} - \adj_\varepsilon$, respectively, we have
    \begin{align}
    \label{eq:lap-diff-bound}
        \norm{\lap - \lap_\varepsilon}_2 \le (e^{\varepsilon}-1) \norm{\lap}_2.
    \end{align}
    \end{enumerate}

\subsection{Feasible Parameter Regime}
\label{sec:sparse-graph-feasible-parameter-regime}

\begin{figure}[!tb]
    \centering
    \includegraphics[width=\linewidth]{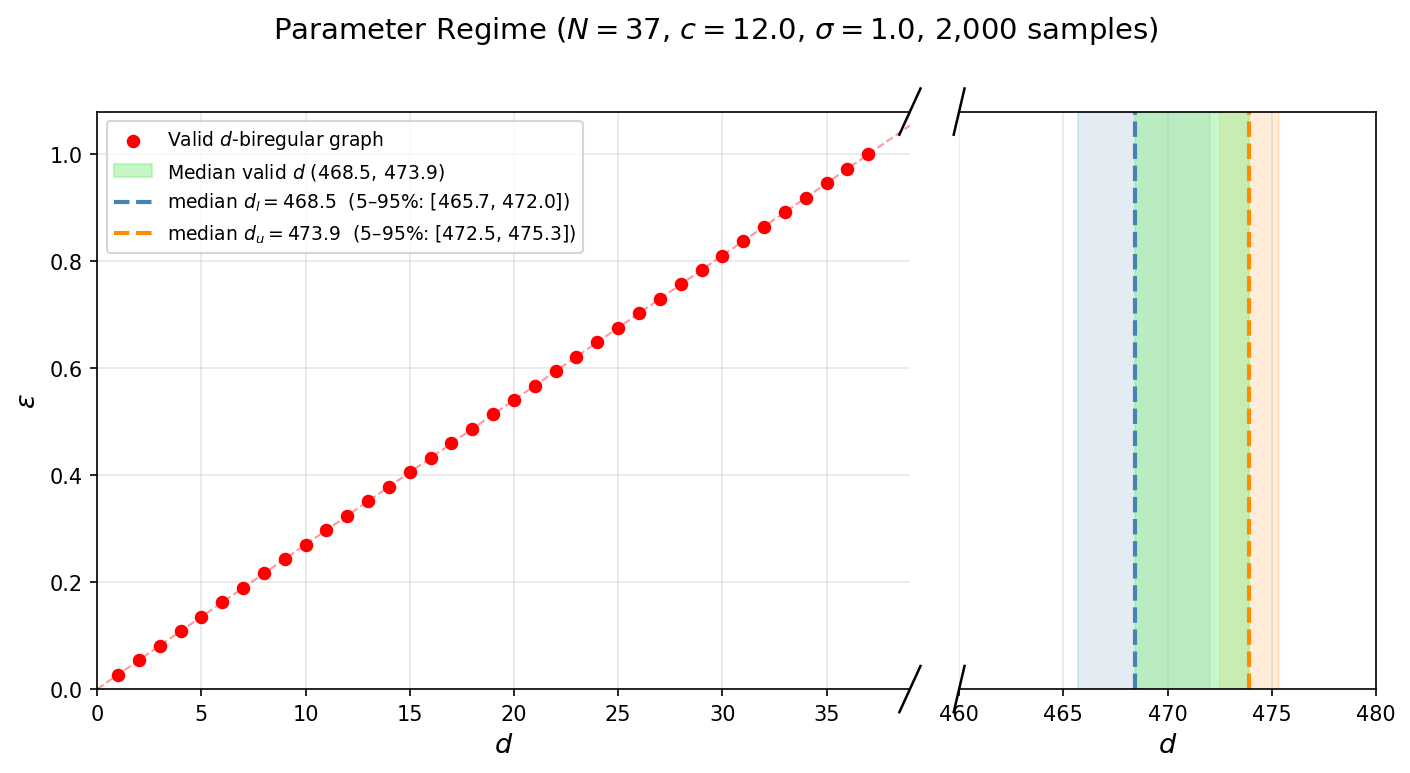}
    \caption{Comparison of parameter regimes for $(37, d)$-BEG-GC (red dots) and $(37, 12, d, \varepsilon)$-EP-GC (green region). For EP-GC to be valid, $d$ must
   satisfy Eq.\eqref{eq:d-bound}. Based on 2,000 independent samples of $\mathbf{E}_0$, the green band spans the median feasible interval; the shaded bands indicate the 5th–95th percentile range of each bound ($d_{\mathrm{l}}$ in blue, $d_{\mathrm{u}}$ in orange).}
    \label{fig:EP-parameter-regime}
\end{figure}
Both $(N,d)$-BEG-GC and $(N,c,d,\varepsilon)$-EP-GC can be viewed as undirected weighted graphs: BEG-GC uses binary weights, while EP-GC uses real-valued weights. Although $d$ denotes the node degree in both, it plays different roles. In BEG-GC, $d$ determines sparsity and thus fixes the workload and replication levels. In EP-GC, $d$ specifies only the graph degree, while sparsity is independently controlled by $\varepsilon$.

Consequently, BEG-GC tightly couples degree, redundancy, and computational load, whereas EP-GC decouples structural regularity from sparsity. This decoupling enables a continuous tradeoff between accuracy and efficiency through $\varepsilon$, without altering $d$, and leads to substantially greater parameter flexibility. In particular, BEG-GC and BIBD-GC constructions impose rigid combinatorial constraints on $(N,d,L,R)$, while EP-GC admits a significantly broader feasible parameter regime, as characterized below.
\begin{theorem}
    An $(N, c, d, \varepsilon)$ EP-GC exists, if and only if 
    \begin{equation}
    \label{eq:d-bound} \underbrace{{\max_{i\in[N-1]}\mathsmaller\sum\limits_{j\in[N-1]}} (\GEP_0)_{i j}}_{d_{\mathrm{l}}}< d<\underbrace{\tfrac{\sum_{i, j\in[N-1]} (\GEP_0)_{ij}}{N-2}}_{d_{\mathrm{u}}}.
    \end{equation}
\end{theorem}
\begin{proof}
   Positivity of all entries in $\GEP$ requires Eqs.~\eqref{eq:append-col} and \eqref{eq:alpha} to be greater than zero, yielding Eq.\eqref{eq:d-bound}.
\end{proof}

\begin{corollary}
% Write $S_i = \sum_{j=1}^{N-1}|X_{ij}|$, $S_{\max} = \max_{i\in[N-1]} S_i$, and
% $S_{\mathrm{tot}} = \sum_{i,j\in[N-1]}|X_{ij}|$.
% Then
% \begin{align}
%     d_{\mathrm{lo}} &= (N-1)C + \sigma S_{\max}, \\
%     d_{\mathrm{hi}} &= \frac{(N-1)^2 C}{N-2} + \frac{\sigma S_{\mathrm{tot}}}{N-2},
% \end{align}
% and the width of the valid interval satisfies
% \begin{equation}
%     d_{\mathrm{hi}} - d_{\mathrm{lo}}
%     = \frac{(N-1)C}{N-2}
%       + \frac{\sigma\!\left(S_{\mathrm{tot}} - (N-2)S_{\max}\right)}{N-2}.
% \end{equation}
The necessary conditions for Eq.\eqref{eq:d-bound} to hold are
\begin{equation*}
    c > \tfrac{\Big[(N-2)\max\limits_{i\in[N-1]} \sum\limits_{j\in[N-1]}|X_{ij}| - \sum\limits_{i,j\in[N-1]}|X_{ij}|\Big]}{N-1}\text{ and } N\geq3.
\end{equation*}
\end{corollary}

\begin{proof}[Proof]
Substituting $(\GEP_0)_{ij} = c + |X_{ij}|$ into the expressions for $d_{\mathrm{l}}$
and $d_{\mathrm{u}}$ from Theorem~1 gives
\begin{align*}
    d_{\mathrm{l}}
        % &= \max_{i\in[N-1]}\sum_{j=1}^{N-1}(C + |X_{ij}|)
         &= (N-1)c +  \max_{i\in[N-1]}\mathsmaller\sum\limits_{j\in[N-1]}|X_{ij}|, \\
    d_{\mathrm{u}}
        % &= \frac{\sum_{i,j\in[N-1]}(C + |X_{ij}|)}{N-2}
         &= \tfrac{(N-1)^2 c}{N-2} + \tfrac{ \sum_{i,j\in[N-1]}|X_{ij}|}{N-2}.
\end{align*}
Since $d_{\mathrm{u}}$ contains the factor $(N-2)^{-1}$, we must have $N \neq 2$. Moreover, $N=1$ is impossible, since in that case Eq.~\eqref{eq:d-bound} is not well defined. Hence, as $N$ is a positive integer, it follows that $N \ge 3$.
We then take the difference,
\begin{align}
    &d_{\mathrm{u}} - d_{\mathrm{l}}\notag\\
    &=\tfrac{ \sum\limits_{i,j\in[N-1]}|X_{ij}|}{N-2} +  \max_{i\in[N-1]}\mathsmaller\sum\limits_{j\in[N-1]}|X_{ij}|\notag\\
    &\qquad+c\left[\tfrac{(N-1)^2 }{N-2} - (N-1)\right]\notag\\
    &= \tfrac{\left(\sum\limits_{i,j\in[N-1]}|X_{ij}| - (N-2)\max\limits_{i\in[N-1]}\sum\limits_{j\in[N-1]}|X_{ij}|\right)}{N-2}+ \tfrac{c(N-1)}{N-2}.
    \label{eq:bound-diff}
\end{align}
Setting Eq.\eqref{eq:bound-diff} to be greater than zero and solving for $c$ complete the proof.
% Setting $d_{\mathrm{u}} - d_{\mathrm{l}} > 0$ and solving for $c$ complete the proof.
\end{proof}
% To ensure that \textsc{DegreePreservingSparsify} produces a valid degree-preserving $\varepsilon$-sparsifier, it suffices that:
% \begin{enumerate}
%     % \item $d$ exceeds the maximum row sum of $\GEP_0$, ensuring positivity of all entries\footnote{In practice, a small constant ($1\times 10^{-9}$) is added to avoid degeneracy.};
%     % \item $\alpha$ is positive;
%     \item every entry in $E_0$ and $\boldsymbol{m}$, and $\alpha$ are positive, i.e.,
    
%     \item $d$ grows at most polynomially in $N$, i.e., $|d| \le N^C$ for some fixed $C>0$, ensuring all entries remain polynomially bounded.
% \end{enumerate}
These conditions are mild and impose no combinatorial structure. Moreover, $\varepsilon \in (0,1]$ can be chosen to continuously tune sparsity. As a result, for fixed $N$, EP-GC admits a significantly broader feasible parameter regime than BEG-GC, as illustrated in Fig.~\ref{fig:EP-parameter-regime}.

\subsection{Theoretical Performance Evaluation}
\label{sec:performance-evalution}

We evaluate the performance of the proposed expansion-preserving gradient code by comparing it to BIBD gradient code and the sparse Gaussian gradient code. Through theoretical analysis, we derive the error performance as demonstrated in the accompanying proof (see supplementary material). The main results are summarized as follows:

\begin{theorem}
\label{thm:EP-GC-error-performance}
Given a weighted bipartite graph $G$, let $G_\varepsilon = $ \textsc{DegreePreservingSparsify}$(G, \varepsilon)$ and $\GEP_\varepsilon$ be the bi-adjacency matrix of $G_\varepsilon$, we have
\begin{align}
         \err{\GEP_\varepsilon} &\leq \tfrac{1}{N}\bigg[  \left(\tfrac{2(e^{\varepsilon}-1)N}{\sqrt{N-S}}\right)+\tfrac{\lambda_2(\GEP)}{d} \sqrt{\tfrac{NS}{N-S}}\bigg]^2
            \label{eq:error-upperbound-by-the-second-largest-lambda}.
\end{align}
 % \label{eq:error-upperbound-without-the-second-largest-lambda}
%  There exists $C_1, C_2>0$ such that for any $\delta >0$, we have
% \begin{align}
% \label{eq:high-prob-bound-for-second-largest-eigenvalue-of-initial-matrix}
%  \lambda_2(\GEP) \leq \big(\sqrt{\tfrac{2}{\pi}}+c\big)N + C_1(\sqrt{N}+\delta)
% \end{align}
% with probability at least $1-2\exp(-C_2\delta^2)$.
\end{theorem}
% \lele{[Did we say that we will add another upper bound in terms of $\frac{\lambda_2(\GEP)}{d}$ in Theorem 3?]}
Our proof of Theorem \ref{thm:EP-GC-error-performance} is as follows.
\begin{proof} We first choose the decoding vector defined in Eq.\eqref{eq:decoding_vec}.
% \begin{equation}
% \label{eq:decoding_vec}
%     \v_\F = \frac{1}{d}(\ones{N}+\a_{\F}),
% \end{equation}
% where
% \begin{equation}
% \label{eq:decoding-vec-component}
%     \a_\F = 
%     \begin{cases}
%         -1, \qquad i\not \in \F\\
%         \frac{S}{N-S} \qquad i\in \F.
%     \end{cases}
% \end{equation}
Therefore, we can bound $\err{\GEP}$ by
\begin{align}
% \label{eq:error-upperbound}
    &\err{\GEP_\varepsilon}
    \le \tfrac{1}{N}\max_{\substack{\F\subset [N]\\ |\F| = N-S}}\big\|\GEP_{\varepsilon}\v_\F - \ones{K}\big\|_2^2\notag\\
     &\quad= \tfrac{1}{N}\max_{\substack{\F\subset [N]\\ |\F| = N-S}}
       \big\|
           \underbrace{\GEP_{\varepsilon}\v_\F - \GEP\v_\F}_{\boldsymbol{\Delta}_1}
           + 
           \underbrace{\GEP\v_\F - \ones{N}}_{\boldsymbol{\Delta}_2}
       \big\|_2^2
     \notag\\
     &\quad\overset{(a)}{\le}
       \tfrac{1}{N}\max_{\substack{\F\subset [N]\\ |\F| = N-S}}
       \left(
           \norm{\boldsymbol{\Delta}_1}_2 + \norm{\boldsymbol{\Delta}_2}_2
       \right)^2,
     \label{eq:error-upperbound-decomposition}
\end{align}
where (a) follows by applying the submultiplicativity of the spectral
norm and the triangle inequality.

We now state the following lemma, with its proof deferred to Appendix \ref{sec:proof-of-lemma-diff-decoded-result-bound}.

\begin{lemma}
\label{lemma:diff-decoded-result-bound}
Given a weighted bipartite graph $G$, let $G_\varepsilon = $ \textsc{DegreePreservingSparsify}$(G, \varepsilon)$. For any nonempty set $\F\subset[n]$ of size $N-S$, with the choice of $\v_\F$ being defined as Eq.\eqref{eq:decoding_vec}, we have
\begin{align*}
        \norm{\boldsymbol{\Delta}_1}_2 \leq  \tfrac{2(e^\varepsilon-1)N}{\sqrt{N-S}}\text{ and } \norm{\boldsymbol{\Delta}_2}_2 &\leq \tfrac{\lambda_2(\GEP)}{d} \sqrt{\tfrac{NS}{N-S}}.
\end{align*}
\end{lemma}
According to Lemma~\ref{lemma:diff-decoded-result-bound}, the bounds of $\boldsymbol{\Delta}_1$ and $\boldsymbol{\Delta}_2$ are irrelevant to the non-straggling pattern $\F$ and thus right-hand-side term of Eq.\eqref{eq:error-upperbound-decomposition} is upper bounded by
% \begin{align*}
%        % &\max_{\substack{\F\subset [N]\\ |\F| = N-S}}
%        % \left(
%        %     \norm{\boldsymbol{\Delta}_1}_2 + \norm{\boldsymbol{\Delta}_2}_2
%        % \right)^2 \notag\\
%        % &\quad= \left(
%        %     \norm{\text{(a)}}_2 + \norm{\text{(b)}}_2
%        % \right)^2 \notag\\
%        &\tfrac{1}{N}\left[2\varepsilon \sqrt{N-S}\left(\tfrac{N}{N-S}\right) + \tfrac{\lambda_2(\GEP)}{d} \sqrt{\tfrac{NS}{N-S}}\right]^2,
%        % \label{eq:error-upperbound-by-the-second-largest-lambda-intermediate}
% \end{align*}
$\tfrac{1}{N}
\left(
   \norm{\boldsymbol{\Delta}_1}_2 + \norm{\boldsymbol{\Delta}_2}_2
\right)^2,$
which completes the proof of Eq.\eqref{eq:error-upperbound-by-the-second-largest-lambda}.
\end{proof}

Since $\GEP$ is random, $\lambda_2(\GEP)$ is random. We will now upper bound $\lambda_2(\GEP)$ by using design parameters $(N,c, d,\varepsilon)$.

We start analyzing $\|\GEP_0\|_2\triangleq \sigma_{1}(\GEP_{0})$, which can be written as
\begin{align}
\label{eq:spectrum-norm-of-intial-matrix}
    \|\GEP_0\|_2 &\overset{(b)}{=} \sqrt{\lambda_{1}({\GEP_{0}}^{\top}\GEP_{0})} \overset{(c)}{=} \sqrt{\lambda_{1}\big((\GEP_{0})^{2}\big)}\notag\\
    &= \sqrt{\big(\lambda_{1}(\GEP_{0})\big)^{2}}= |\lambda_1(\GEP_0)|
\end{align}
where (b) is because the square of each singular value of square matrix $\GEP_0$ equals a non-negative eigenvalues of ${\GEP_0}^\top \GEP_0$; (c) is because $\GEP_{0}$ is symmetric.
% By definition, the square of each singular value of square matrix $\GEP_0$ equals a non-negative eigenvalues of ${\GEP_0}^\top \GEP_0$. Therefore, the spectral norm can be written as
% \begin{align*}
%     \|\GEP_0\|_2 \triangleq \sigma_{1}(\GEP_{0}) = \sqrt{\lambda_{1}({\GEP_{0}}^{\top}\GEP_{0})}.
% \end{align*}
% Because \(\GEP_{0}\) is symmetric, we have ${\GEP_{0}}^{\top} = \GEP_{0}$, which means ${\GEP_{0}}^{\top}\GEP_{0} = (\GEP_{0})^{2}$. So, we have
% The eigenvalues of $(\GEP_{0})^{2}$ are precisely the squares of the eigenvalues of $\GEP_{0}$.  Thus,
% \begin{align}
%     \norm{\GEP_0}_2 &= \sqrt{\lambda_{1}\big((\GEP_{0})^{2}\big)}= \sqrt{\big(\lambda_{1}(\GEP_{0})\big)^{2}}= |\lambda_1(\GEP_0)|.\notag
%     % \label{eq:spectrum-norm-of-sparse-graph-scheduling-matrix}
% \end{align}
Following the construction given in Sec.\ref{sec:sparse-graph-gradient-code-encoding-matrix-construction}, $\GEP$ is random, and thus $\lambda_2(\GEP)$ in \eqref{eq:error-upperbound-by-the-second-largest-lambda} is a random variable. We therefore derive a high-probability upper bound on $\lambda_2(\GEP)$ in terms of the design parameters $(N,c, d,\varepsilon)$.
\begin{theorem}
For the matrix $\GEP$ constructed in Sec.~\ref{sec:sparse-graph-gradient-code-encoding-matrix-construction}, there exists $C_1>0$ such that for any $t >0$, we have
\begin{align}
% \label{eq:lambda2-tail-bound}
    \lambda_2(\GEP) \leq |\mu+c|(N-1) + C_1(\sqrt{N-1}+t)\notag
\end{align}
with probability at least $1-4e^{-t^2}$.
\end{theorem}
\begin{proof}
We leverage the following Lemma to bound $\lambda_2(\GEP)$.
\begin{lemma} [Cauchy's eigenvalue interlacing theorem{\cite[Theorem 4.3.15]{Horn_Johnson_1985}}]
\label{lem:eigenvalue-interlacing-theorem}
Suppose $\boldsymbol{A}\in \RR^{n\times n}$ is symmetric. Let $\boldsymbol{B}\in \RR^{{(n-1
)}\times {(n-1)}}$ be a principal submatrix (obtained by deleting both $i$-th row and $i$-th column for some values of $i$). Suppose $\boldsymbol{A}$ has eigenvalues $\lambda_1\geq \dots \geq \lambda_n$ and $\boldsymbol{B}$ has eigenvalues $\beta_1\geq \dots \geq \beta_m$. 
Then
% \[
% \lambda_k\geq \beta_k\geq \lambda_{k+n-m} \quad\text{for} \quad k=1, \dots, m
% \]
% And if $m = n-1$,
\[
    \lambda_1 \geq \beta_1 \geq \lambda_2 \geq \beta_2 \geq \dots \geq \beta_{n-1}\geq \lambda_n.
\]
\end{lemma}
Recall our construction of $\GEP$ as of Eq.\eqref{eq:extended-sparse-graph-scheduling-matrix}, we can rewrite $\GEP_0$ and apply Lemma~\ref{lem:eigenvalue-interlacing-theorem} and conclude $\lambda_2(\GEP) \geq \lambda_2(\GEP_0).$
Therefore, combined with Eq.\eqref{eq:spectrum-norm-of-intial-matrix}, we obtain
\begin{align}
 \label{eq:the-second-largest-eigenvalue-of-sparse-graph-scheudling-matrix}
    \lambda_2(\GEP) &\leq \lambda_2(\GEP_0) \leq |\lambda_1(\GEP_0)| =  \|\GEP_0\|_2
\end{align}

Then, we turn to bound $\|\GEP_0\|_2$ by using the following Lemma, whose proof is deferred to Appendix \ref{proof of lem:norm-of-matrices-with-standard-gaussian-entries}.
% \begin{lemma} [{Modified\cite[Theorem 4.4.5]{vershynin2020high}}]
% \label{lem:norm-of-matrices-with-standard-gaussian-entries}
% Let $\boldsymbol{M}$ be an $n\times n$ symmetric matrix, for $1\le i\le j\le N-1$, $\boldsymbol{M}_{ij} = |G_{ij}|$, where $G_{ij}~\overset{\mathrm{i.i.d.}}{\sim}~\gaussian{0}{1}$. Then, there exists absolute constants $C_1, C_2$ such that for any $\varepsilon >0$, we have
% \[
%     \norm{\boldsymbol{M}}_2\leq \sqrt{\tfrac{2}{\pi}}n +C_2(\sqrt{n}+\varepsilon)
% \]
% with probability at least $1-2\exp(-C_1\varepsilon^2)$.
% \end{lemma}
\begin{lemma}
\label{lem:norm-of-matrices-with-standard-gaussian-entries}
For the matrix $\GEP_0$ constructed in Sec.~\ref{sec:sparse-graph-gradient-code-encoding-matrix-construction}, there exist absolute constants
$C_1>0$ such that for any $t>0$,
\[
% \Pr\Big(
\|\GEP_0\|_2
\le
|\mu+c|(N-1) + C_1(\sqrt{N-1}+t)
% \Big)
% \ge
% 1-2e^{-C_2\delta^2 (N-1)}.
\]
with probability at least $1-4e^{-t^2}.$
\end{lemma}
% \begin{lemma}
% \label{lem:norm-of-matrices-with-standard-gaussian-entries}
% Let $\boldsymbol{M}$ be an $n\times n$ symmetric matrix such that
% % \[
% $\boldsymbol{M}_{ij} = |G_{ij}| + C,\quad G_{ij}\overset{\mathrm{i.i.d.}}{\sim}\gaussian{0}{1},$
% % \]
% for $1\le i\le j\le n$, and $\boldsymbol{M}_{ji}=\boldsymbol{M}_{ij}$ for $i>j$. Then there exist absolute constants $C_1,C_2>0$ such that for any $\varepsilon>0$,
% \[
% \Pr\Big(
% \|\boldsymbol M\|_2
% \le
% \big(\sqrt{\tfrac{2}{\pi}} + C\big)n + (C_1+\varepsilon)\sqrt n
% \Big)
% \ge
% 1-2e^{-C_2\varepsilon^2 n}.
% \]
% \end{lemma}
% The proof of Lemma \ref{lem:norm-of-matrices-with-standard-gaussian-entries} is deferred to Appendix \ref{proof of lem:norm-of-matrices-with-standard-gaussian-entries}.

% Since the entries of $\GEP_0$ follow i.i.d.\! half-normal distribution (cf.\ref{sec:sparse-graph-gradient-code-encoding-matrix-construction}), we apply Lemma \ref{lem:norm-of-matrices-with-standard-gaussian-entries} and thus we obtain that, there exists absolute constants $C_1>0$ such that for any $\delta>0$
% \begin{align*}
% \label{eq:concentration-spectrum-norm-of-sparse-graph-scheduling-matrix}
%     \prob{\norm{\GEP_0}_2\geq \big(\sqrt{\tfrac{2}{\pi}}+c\big)N + C_1\sqrt{N}} \leq 2\exp{(-C_2\delta^2N)}.
% \end{align*}
% Together with Eq.\eqref{eq:the-second-largest-eigenvalue-of-sparse-graph-scheudling-matrix}, we complete the proof.\qedhere
Combining Lemma~\ref{lem:norm-of-matrices-with-standard-gaussian-entries} with
Eq.\eqref{eq:the-second-largest-eigenvalue-of-sparse-graph-scheudling-matrix}
completes the proof.\qedhere
% , we conclude that there exists $C_0, C>0$ such that for any $\delta >0$, with probability at least $1-2\exp(-C_0\delta^2)$,
% \begin{align}
%     \lambda_2(\GEP)
%     &\leq  \sqrt{\tfrac{2}{\pi}}N + C(\sqrt{N}+\delta).
%     \label{eq:lambda2-tail-bound}
% \end{align}

% With the combination of 
% % Eq.\eqref{eq:error-upperbound}, 
% \eqref{eq:error-upperbound-decomposition}
% % , \eqref{eq:error-upperbound-by-the-second-largest-lambda-intermediate} 
% and \eqref{eq:lambda2-tail-bound}, we complete the proof.
\end{proof}

\subsection{Parameter Selection and Error Trade-offs}
As indicated by Eq.~\eqref{eq:error-upperbound-by-the-second-largest-lambda}, the error is governed by the competition between the sparsification and the spectral terms:
\[
T_{\mathrm{sp}} \triangleq \tfrac{2(e^{\varepsilon}-1)N}{\sqrt{N-S}},
\quad
T_{\mathrm{spec}} \triangleq \tfrac{\lambda_2(\GEP)}{d}\sqrt{\tfrac{NS}{N-S}}.
\]
The sparsification term dominates when 
\[
\tfrac{\lambda_2(\GEP)}{d} \ll 2(e^\varepsilon-1)\sqrt{\tfrac{N}{S}},
\]
and the spectral term dominates otherwise.
\begin{corollary}
If an $(N, c, d, \varepsilon)$ EP-GC exists, then $ d\geq\lambda_1(\GEP_0) \geq \lambda_2(\GEP)$.
\end{corollary}
\begin{proof}
By Lemma~\ref{lem:biadj-matrix}, we have $d = \sigma_1(\GEP) \ge \lambda_1(\GEP)$.
Combined with Eq.~\eqref{eq:the-second-largest-eigenvalue-of-sparse-graph-scheudling-matrix}, the statement holds.
\end{proof}
By the corollary, any feasible $(N,d,\varepsilon)$ EP-GC satisfies $\tfrac{\lambda_2(\GEP)}{d} \le 1$ and hence $T_{\mathrm{spec}} = O\left(\sqrt{\tfrac{NS}{N-S}}\right)$.

If $S=o(N)$, then $T_{\mathrm{spec}}=O(\sqrt{S})$ while $T_{\mathrm{sp}}=\Theta(\sqrt{N})$ for fixed $\varepsilon>0$, so the sparsification term dominates and the error is 
% primarily 
controlled by $e^\varepsilon-1$. If $S=\Theta(N)$, both terms are $\Theta(\sqrt{N})$, and the error depends jointly on $\tfrac{\lambda_2(\GEP)}{d}$ and $\varepsilon$.

Therefore, for parameter selection, $d$ should be chosen above $\lambda_1(\GEP_0)$ to ensure feasibility and suppress the spectral term. Beyond this threshold, further increasing $d$ yields diminishing returns, and the dominant tradeoff is governed by $\varepsilon$. Specifically, \textsc{DegreePreservingSparsify} returns a graph with $N^{1+o(1)}\varepsilon^{-2}$ edges~\cite{degree-preserving-sparsifier}: smaller $\varepsilon$ retains more edges and leads to lower error, whereas larger $\varepsilon$ yields a sparser encoding matrix at the cost of increased approximation error.

\section{Experimental Performance Evaluation}
\begin{figure}[!tb]
    \centering
    \includegraphics[width=\linewidth]{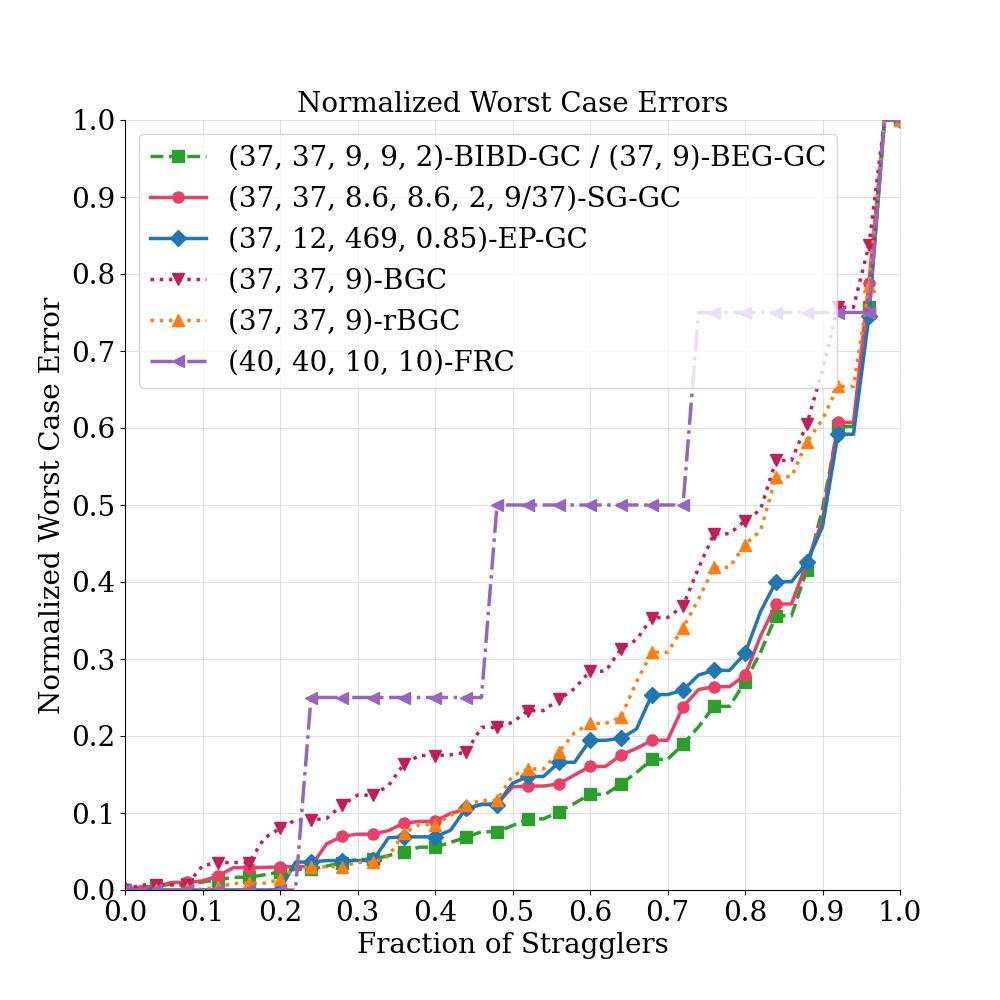}
    \caption{A comparison of worst-case normalized errors between our proposed sparse Gaussian gradient code (SG-GC), expansion-preserving gradient code (EP-GC), BIBD gradient code (BIBD-GC) and other exisiting gradient codes, with a specific focus on similar column densities, being 0.25 for EP-GC, and $0.24$ for SG-GC, BIBD-GC and the other gradient codes.}
    \label{fig:gradient-code-err-performance-comparison}
\end{figure}
In this section, we complement our theoretical analyses with a simulation comparing the worst-case error among six gradient codes with similar workload densities. We select a valid BIBD-GC and choose parameters for other gradient codes to match its column density as closely as possible. To obtain the EP-GC and SG-GC configurations whose performance is closest to that of BIBD-GC, we evaluated thousands of random seeds and selected the best-performing instances for plotting.

As shown in Fig.~\ref{fig:gradient-code-err-performance-comparison}, the BIBD gradient code serves as the baseline due to its near-optimal worst-case reconstruction error across the entire range of straggler fractions. It consistently achieves the lowest error and provides a reference for evaluating other constructions. 

The proposed SG-GC and EP-GC closely track the performance of BIBD-GC over a wide range of straggler fractions. In the low-straggler regime, all three codes exhibit very small errors, indicating that the redundancy structure is sufficient to ensure accurate recovery. As the straggler fraction increases to the mid regime, a modest gap emerges between the proposed methods and BIBD-GC. However, this gap remains small, and both SG-GC and EP-GC preserve a similar growth trend as BIBD-GC.

In the high-straggler regime, the performance of all codes degrades rapidly due to the severe loss of information. Notably, EP-GC becomes nearly indistinguishable from BIBD-GC in this region, suggesting that the expansion-preserving structure effectively captures the dominant spectral properties governing worst-case error. SG-GC also remains competitive, though with slightly larger deviations.

In contrast, other existing gradient codes such as BGC and rBGC exhibit consistently higher errors across all regimes, indicating weaker robustness to stragglers. The FRC construction shows the poorest performance, with pronounced stepwise behavior due to its coarse combinatorial structure, leading to significant error jumps as the number of stragglers increases.

Overall, while BIBD-GC remains the optimal benchmark, the proposed SG-GC and EP-GC achieve near-optimal performance with substantially greater flexibility in parameter selection. This makes them practical surrogates in settings where BIBD constructions are unavailable or infeasible.

\section{Conclusion}
In this paper, we proposed two gradient codes, termed sparse Gaussian gradient code and expansion-preserving gradient code. Both codes are formulated as probabilistic models that employ distinct sparsification schemes. Their error performance is analyzed theoretically and validated through numerical experiments. From both perspectives, they achieve performance comparable to the BIBD gradient code with an appropriate choice of system parameters.

Future work may explore refined sparsification mechanisms that account for intrinsic worker characteristics, such as computation speed and recovery behavior following failures. Extending these constructions to adaptive or dynamically evolving graph structures could further illuminate their effectiveness and scalability in large-scale distributed learning systems.

{\appendices

\section{Proof for Sparse Gaussian Gradient Code}
In this section, we prove Theorem~\ref{thm:SG-GC-error-performance}. We begin with a definition and the required lemmas.
\begin{definition}[{\cite[Definition 2.7.5]{vershynin2020high}}]
\label{def:sub-exponential}
    The sub-exponential norm of a random variable $Z$ is defined as
    \[
        \psione{Z} = \inf\{t>0:\ex{\exp(|Z|/t)\le 2}\}.
    \]
    If $\psione{Z}$ is finite, we say that $Z$ is sub-exponential.
\end{definition}

\begin{lemma}
\label{lemma:psi-two-upper-bound}
 Let ${\sigmamax} = \rho(N-S)^{1/2}(b+(N-S-1)c)^{1/2}$ and $\betamax=~\rho(N-S)a-1$. If $\sigma_{\max} < \sqrt{2}, c\ge0$ and $\rho \ge0$, then for any $i \in [K]$, $Z^2$ is a sub-exponential random variable with sub-exponential norm
    \[\psione{Z^2} \le \max\left\{2\sigmamax^2 + \tfrac{\betamax^2}{\ln{(\sqrt{2}/\sigmamax)}}, 4\sigmamax^2 \right\}.\]
     % where $\betamax=~\rho(N-S)a-1$.
\end{lemma}

\begin{lemma} If $Z$ is a sub-exponential random variable then $Z-\ex{Z}$ is a sub-exponential with
\label{lemma:centering-a-sub-exponential-random-variable}
\[
    \psione{Z-\ex{Z}} \leq C\psione{Z},
\]
where $C>0$ is an absolute constant.
\end{lemma}

\begin{lemma}[Bernstein's inequality {\cite[Theorem 2.8.1]{vershynin2020high}}]
\label{lem:bernstein's-inequality}
    Let $Z_1, \dots, Z_K$ be independent, mean zero, sub-exponential random variables. Then, for every $t>0$, we have
    % \begin{align*}
    %     &\prob{\left|\textstyle{\textstyle\sum_{i=1}^K Z_i\ge t}\right|}\le 2\exp\bigg[
    %         -C\min\Big\{
    %             \tfrac{t^2}{\sum_{i=1}^K\psione{Z_i}^2},
    %             \tfrac{t}{\max_i\psione{Z_i}}
    %         \Big\}
    %     \bigg],
    % \end{align*}
    \begin{align*}
        &\prob{\left|\textstyle{\textstyle\sum_{i=1}^K Z_i\ge t}\right|}\le 2\exp\Big[
            -C\min\big\{
                \tfrac{t^2}{\sum_{i=1}^Kq^2},
                \tfrac{t}{\max_iq}
            \Big\}
        \Big],
    \end{align*}
    where $q = \psione{Z}$ and $C>0$ is an absolute constant.
\end{lemma}

\begin{lemma}
\label{lemma:expectation-of-two-norm-of-X}
Let $\GSG$ be the encoding matrix for an $(N, K, L, R, \gamma)$-SG-GC and $\rho = \frac{L}{L+\lambda(N-S-1)}$. For any non-straggling set $\F$ with $|\F| =~N-~S$, the expected error using constant decoding vector $\rho\ones{\F}$ is given by
  % \[
  %   \ex{\norm{\rho\boldsymbol{E}\ones{\F} - \ones{K}}_2^2} =\textstyle{\ex{\sum_{i=1}^K X_i^2} }=K\cdot\err{\GB},
  %   \]
  \[
    \mathbb{E}\big[{\|\rho\GSG\ones{\F} - \ones{K}}\|_2^2\big]=K\cdot\err{\GB},
    \]
     where $\err{\GB}$ is the error of an $(N, K, L, R, \lambda)$-BIBD-GC.
\end{lemma}
The proofs of Lemmas~\ref{lemma:psi-two-upper-bound}, \ref{lemma:centering-a-sub-exponential-random-variable}, and \ref{lemma:expectation-of-two-norm-of-X} are provided in Appendices~\ref{proof of lemma:psi-two-upper-bound}, \ref{proof of lemma:centering-a-sub-exponential-random-variable}, and \ref{proof of lemma:expectation-of-two-norm-of-X}, respectively.

\subsection{Proof of Theorem \ref{thm:SG-GC-error-performance}}
\begin{proof}
\label{proof of thm:SG-GC-error-performance}
Suppose an $(N, K, L, R, \lambda, \gamma)$-SG gradient code exists and
% Let $\GSG_\F$ be its encoding matrix with a non-straggling set $\F\subset[N]$ and $|\F| = N-S$. For $i\in [K]$, let
% \begin{align}
% \label{eq:definitons-of-X_i-and-Y_i}
% X_i = \rho \textstyle\sum\limits_{j\in \F}^{N}{\GSG_{ij}} -1\quad \text{and}\quad Y_i = X_i^2 - \ex{X_i^2}.    
% \end{align}
denote its encoding matrix as $\GSG$. For any non-straggling set $\F\subset[N]$ with $|\F|=N-S,$ $ i\in[K]$, we define
\begin{align}
X_i = \rho \textstyle\sum_{j\in \F} \GSG_{ij} -1,\quad 
Y_i = X_i^2 - \ex{X_i^2}.
\end{align}

Since rows in $\GSG$ are sampled from the same multivariate Gaussian distribution, i.e. $\gaussian{\muu}{\Sigmaa}$, independently (cf. Sec.\ref{sec:SG-GC-construction}), then $Y_1, \dots, Y_K$ is i.i.d. with mean zero. 
Under the assumptions on $\rho = \frac{L}{L+\lambda(N-S-1)}>0$, $\rho^2(N-S)(b+(N-S-1)c)<2$ and $c\geq 0$, Lemma~\ref{lemma:psi-two-upper-bound} gives
\begin{equation*}
    \psione{X_i^2}\le \max\bigg\{2\sigmamax^2 + \tfrac{\betamax^2}{\ln{(\tfrac{\sqrt{2}}{\sigmamax})}}, 4\sigmamax^2\bigg\},
\end{equation*}
where $2\sigmamax^2 + \betamax^2\frac{1}{\ln{(\frac{\sqrt{2}}{\sigmamax})}} = \Theta (1)$ and $4\sigmamax^2 = \Theta(1)$.
So, 
% $\psione{X_i^2}$ is finite, and thus
$X_i^2$ is a sub-exponential random variable by Definition~\ref{def:sub-exponential}.
% The proof of Lemma~\ref{lemma:centering-a-sub-exponential-random-variable} can be found in Appendix \ref{proof of lemma:centering-a-sub-exponential-random-variable}.

From Lemma~\ref{lemma:centering-a-sub-exponential-random-variable}, $\psione{Y_i}\le C_1\psione{X_i^2}$ for a constant $C_1>0$, and thus $Y_i$ is also a sub-exponential random variable according to Lemma~\ref{lemma:psi-two-upper-bound}.
Similarly, according to Lemma~\ref{lemma:psi-two-upper-bound}, $Y_i$ is also a sub-exponential random variable and thus
\begin{align*}
    \psione{Y_i}&\le C_1 \psione{X_i^2}\le C_1\max\bigg\{2\sigmamax^2 + \tfrac{\betamax^2}{\ln{(\tfrac{\sqrt{2}}{\sigmamax})}}, 4\sigmamax^2\bigg\}.
\end{align*}
Since rows of $\GSG$ are i.i.d.\ Gaussian~(cf. Sec.\ref{sec:SG-GC-construction}), $\{Y_i\}_{i=1}^K$ are i.i.d.\ mean-zero. We apply Lemma \ref{lem:bernstein's-inequality} on  $\{Y_i\}_{i=1}^K$ and choose $t = K^{1/2+\delta}$, where $0<\delta<1/2$.
So, for any $\F\subset[N]$ with $|\F| = N-S$, we have 
\begin{align}
    &\prob{\left|\textstyle\sum_{i=1}^{K}Y_i\right|\ge K^{1/2+\delta}}\notag\\
    &\qquad \le 2\exp\bigg[
        -C_2 \min\Big(
            \tfrac{K^{1+2\delta}}{K\cdot C_1\cdot \psione{X_i^2}^2},
            \tfrac{K^{1/2+\delta}}{\psione{X_i^2}}
        \Big)
        \bigg] \label{eq:Y-concentration},
\end{align}
where $C_2>0$. Thus,
for any $\F\subset[N]$ with $|\F| = N-S$,
\begin{align}
    &\prob{\big|\err{\GSG} - \mathbb{E}[{\err{\GSG}}]\big| \geq K^{1/2+\delta}}\notag\\
        &\qquad\overset{\mathrm{(a)}}{=}\prob{\left|\textstyle\sum_{i=1}^{K}X_i - \sum_{i=1}^{K}\ex{X_i}\right|\geq K^{1/2+\delta}}\notag\\
        &\qquad=\prob{\left|\textstyle\sum_{i=1}^{K}\left(X_i - \ex{X_i}\right)\right|\geq K^{1/2+\delta}}\notag\\
        &\qquad= \prob{\left|\textstyle\sum_{i=1}^{K}Y_i\right|\geq K^{1/2+\delta}}
        \notag\\
        &\qquad\overset{(b)}{\le}2\exp\bigg[
        -C_2 \min\Big(
            \tfrac{K^{2\delta}}{ C_1\cdot\psione{X_i^2}^2},
            \tfrac{K^{1/2+\delta}}{\psione{X_i^2}}
        \Big)
        \bigg]
        \\
        &\qquad=2\exp\bigg[
        -C_2 \min\left(
            \tfrac{K^{2\delta}}{C_1\psione{X_i^2}^2},
            \tfrac{K^{1/2+\delta}}{\psione{X_i^2}}
        \right)
        \bigg], 
        \label{eqn:prob-existence}
\end{align}
where an absolute constant $C_2>0$. Above, (a) and (b) follow from Lemma~\ref{lemma:expectation-of-two-norm-of-X} and Eq.~\eqref{eq:Y-concentration}, respectively. Therefore, we have
\begin{align}
    &\prob{\err{\GSG}\leq \err{\GB} + K^{\delta-1/2}} \notag\\
    &= \prob{K\cdot \err{\GSG}\leq K\err{\GB} + K^{1/2+\delta}}\notag\\
    &\overset{\mathrm{(c)}}{\ge} \mathbb{P}\bigg[\max_{\substack{\F\subset[N]\\|\F| = N-S}} \err{\GSG}  \leq K\cdot\err{\GB} + K^{1/2+\delta}\bigg]\notag\\
    &\overset{\mathrm{(d)}}{=} 
    \prob{\substack{
    \forall \F\subset [N], |\F| = N-S,\\
    \err{\GSG}  \le \ex{\err{\GSG}} + K^{1/2+\delta}}}\notag\\
    &\ge\prob{\substack{
    \forall \F\subset [N], |\F| = N-S,\\
    \big|\err{\GSG} - \ex{\err{\GSG} }\big| < K^{1/2+\delta}}}\notag\\
    &= 1-\prob{\substack{
    \exists \F, |\F| = N-S, \text{ s.t. }\\
    \big|\err{\GSG}  - \ex{\err{\GSG} }\big| \ge K^{1/2+\delta}}}\notag\\
    &\overset{\mathrm{(e)}}{\ge} 1- 2\tbinom{N}{S}\prob{\left|\err{\GSG} - \ex{\err{\GSG} }\right| \geq K^{1/2+\delta}}
    \label{eq:error-for-all}
\end{align}
Above, (c) follows since 
\begin{align*}
    & K\cdot\err{\GSG} \triangleq \max_{\substack{\F\subset[N]\\|\F| = N-S}}
\|\GSG_{\F}\vopt - \ones{K}\|_2^2 \\
&\le \max_{\substack{\F\subset[N]\\|\F| = N-S}}
\|\rho\GSG_{\F}\ones{N-S}-\ones{K}\|_2^2 =\max_{\substack{\F\subset[N]\\|\F| = N-S}} \errc.
\end{align*}
(d), (e) 
% and (f) 
follow from Lemma~\ref{lemma:expectation-of-two-norm-of-X}, the union bound, 
% and Eq. \eqref{eqn:prob-existence}, 
respectively.

% Recall that
% $
%     a = \frac{L}{K\gamma},
%     b = (\frac{L\gamma}{K}-\frac{L^2}{K^2})\frac{1}{\gamma},
%     c = (\frac{\lambda}{K}-\frac{L^2}{K^2})\frac{1}{\gamma^2},
%     \rho = \frac{L}{L+\lambda(N-S-1)},
%     {\sigmamax}^2 = \rho^2((N-S)b+((N-S)^2-(N-S))c)$ and $\betamax = \rho(N-S)a-1$.
Recall system parameters satisfy Eq.\eqref{eq:sys-para}. With the assumptions of $\gamma = \Theta(1), K = \Theta(N), L = \Theta(N)$ and $\lambda = \Theta(N)$, we have
$ a = \Theta(1),
    b = \Theta(1),
    c = \Theta(1).$
Moreover, since $S = O(N^{\alpha})$ for $0<\alpha<2\delta$, we obtain
 $\rho = \Theta(1/N),
    \sigmamax^2 = \Theta(1),
    \betamax = \Theta(1).$
Thus, $2\sigmamax^2 + \frac{\betamax^2}{\ln{(\frac{\sqrt{2}}{\sigmamax})}} = \Theta(1)$
and $4\sigmamax^2 = \Theta(1).$
Furthermore, $\psione{X_i^2} = \Theta(1).$
% \begin{equation*}
%     \psione{X_i^2} = \Theta(1).
% \end{equation*}
 It follows that $K^{2\delta}/\psione{X_i^2}^2 =\Theta(N^{2\delta})$ and $K^{1/2+\delta}/\psione{X_i^2} =\Theta(N^{1/2+\delta})$. Since when $0<\delta<1/2$,  $N^{1/2+\delta}$ grows faster than $N^{2\delta}$, we obtain 
\begin{equation}
\min\Big(
            \tfrac{K^{2\delta}}{\psione{X_i^2}^2},
            \tfrac{K^{1/2+\delta}}{\psione{X_i^2}}
        \Big) = \Theta(N^{2\delta}).
\label{eq:min-order}
\end{equation}
Thus, combined with Eqs.\eqref{eq:Y-concentration}, \eqref{eqn:prob-existence}, \eqref{eq:error-for-all} and \eqref{eq:min-order}, we have
\begin{align*}
    &\mathbb{P}\big[\err{\GSG}\leq \err{\GB} + K^{\delta-1/2}\big]\\
    &\qquad\ge 1-2\tbinom{N}{S}\exp[
        -\Theta(N^{2\delta})
    ].
\end{align*}
Note that
$\binom{N}{S} \le N^S = \exp(S\ln N)$. When $S = O(N^\alpha)$ and $0 < \alpha < 2\delta$, we have
\begin{align*}
\tbinom{N}{S}\exp(-\Theta(N^{2\delta})) 
&= \exp(-\Theta(N^{2\delta})+O(N^{\alpha}\ln N))\\
&= \exp(-\Theta(N^{2\delta})) = o(1),
\end{align*}
which completes the proof of Theorem~\ref{thm:SG-GC-error-performance}.
\end{proof}

\section{Proof for Lemmas}
% We supplement proof details for all the Lemmas.

\subsection{Proof of Lemma \ref{lemma:diff-decoded-result-bound}}
\label{sec:proof-of-lemma-diff-decoded-result-bound}
\begin{proof}
\label{proof of lemma:diff-decoded-result-bound}
    % We aim to upper bound term (a) and (b) in Eq.\eqref{eq:error-upperbound-decomposition}, separately.
    With the choice of decoding vector being defined by Eq.\eqref{eq:decoding_vec}, we first bound $\norm{\boldsymbol{\Delta}_1}_2$ in Eq.\eqref{eq:error-upperbound-decomposition} by arranging it as
    \begin{align}
    \|\boldsymbol{\Delta}_1\|_2 &=\|\GEP_{\varepsilon}\v_\F-\GEP\v_\F \|_2\notag\\ 
    &= \frac{1}{d} \|(\GEP_{\varepsilon} - \GEP)(\ones{N}+\a_\F)\|_2\notag\\
    \label{eq:first-term}
    &\overset{(a)}{\leq} \frac{1}{d}\|\GEP_{\varepsilon} - \GEP\|_2 \norm{\ones{N} + \a_\F}_2,
\end{align}
where (a) follows from the sub-multiplicative property of a matrix norm.
To complete the bounding for Eq.\eqref{eq:first-term}, we first finish the calculation for $\norm{\ones{N} + \a_\F}_2$ first as follows:
\begin{align}
    \norm{\ones{N} + \a_\F}_2 &= \sqrt{\textstyle\sum_{i\not\in \F} (\ones{N} + \a_\F)_i^2 +  \sum_{i\in \F} (\ones{N} + \a_\F)_i^2}\notag\\
    &= \sqrt{\textstyle\sum_{i\not \in \F} (1-1)^2 + \sum_{i\in \F} \left(1+\frac{S}{N-S}\right)^2}\notag\\
    \label{eq:decodingvec-l2norm}
    &= \sqrt{N-S}\left(1+\tfrac{S}{N-S}\right).
\end{align}
Then, we bound $\norm{\GEP_{\varepsilon} - \GEP}_2$. Let $\adj$ and $\adj_\varepsilon$ be the adjacency matrices of $G$ and $G_\varepsilon$ respectively. We have
\begin{align*}
    \adj - \adj_\varepsilon 
    &= \begin{bmatrix}
        \allzeros{} & \GEP \notag\\
        {\GEP}^\top & \allzeros{} 
    \end{bmatrix} - \begin{bmatrix}
        \allzeros{} & \GEP_\varepsilon\\
        {\GEP_{\varepsilon}}^\top & \allzeros{}
    \end{bmatrix}     \notag\\
    &= \begin{bmatrix}
        \allzeros{} & \GEP - \GEP_\varepsilon \\
        {\GEP}^\top-{\GEP_\varepsilon}^\top &\allzeros{}
    \end{bmatrix} \notag\\
    &= \begin{bmatrix}
        \allzeros{} & \GEP - \GEP_\varepsilon \\
        (\GEP-\GEP_\varepsilon)^\top &\allzeros{}
    \end{bmatrix}.
\end{align*}
To bound $\norm{\GEP_{\varepsilon} - \GEP}_2$, it suffices to bound $\norm{\adj - \adj_\varepsilon}_2$. By the following lemma, the singular values of $\GEP$ coincide with the absolute eigenvalues of the associated bipartite adjacency matrix, so any spectral-norm bound for $\GEP_{\varepsilon} - \GEP$ transfers directly to $\norm{\adj-\adj_{\varepsilon}}_2$.
\begin{lemma} [Modified{\cite[Theorem 7.3.3]{Horn_Johnson_1985}}]
\label{lem: eig-adj-sig-biadj} 
    Let $\boldsymbol{B}$ be a real square matrix with singular values $\sigma_1(\boldsymbol{B}) \geq \sigma_2(\boldsymbol{B}) \geq \cdots \geq \sigma_n(\boldsymbol{B}) \geq 0$. For a block matrix $\boldsymbol{P}~\triangleq~\begin{bmatrix}
        \allzeros{} & \boldsymbol{B}\\
        \boldsymbol{B}^\top & \allzeros{}
    \end{bmatrix}$ with eigenvalues $\lambda_1(\boldsymbol{P})~\geq~\lambda_2(\boldsymbol{P})~\geq~\dots~\geq~\lambda_{2n}(\boldsymbol{P})$, we have
    \[
        \{\lambda_i(\boldsymbol{P})\}_{i=1}^{2n} = \{\sigma_i(\boldsymbol{B})\}_{i=1}^n \cup \{-\sigma_i(\boldsymbol{B})\}_{i=1}^n.
    \]
\end{lemma}%
We then obtain that
\begin{align}
\label{eq:same-matrix-norm}
    \norm{\GEP - \GEP_\varepsilon}_2 
    &\overset{(a)}{=} \sigma_1(\GEP - \GEP_\varepsilon) \overset{(b)}{=} \lambda_1(\adj - \adj_\varepsilon)\notag\\
    &\overset{(c)}{=} \sigma_1(\adj - \adj_\varepsilon) \overset{(a)}{=}\norm{\adj - \adj_\varepsilon}_2,
\end{align}
where (a) is because of the definition of matrix norm; (b) is because of Lemma \ref{lem: eig-adj-sig-biadj}; (c) is because of $\adj - \adj_\varepsilon$ is symmetric.

Now, we try to bound $\norm{\adj - \adj_\varepsilon}_2$. We will need to leverage the definition of Laplacian matrix.
\begin{definition} [Laplacian matrix]
\label{def:laplacian-matrix}
    For a weighted graph $G$ with $n$ nodes, the Laplacian matrix $\lap_{n\times n}$ is defined as 
    $
    \lap = \deg - \adj,
    $
    where $\deg$ is the degree matrix, and $\adj$ is the graph's weighted adjacency matrix.
\end{definition}

We use the Definition \ref{def:laplacian-matrix} and triangle inequality, and obtain 
\begin{align}
    \norm{\adj - \adj_\varepsilon}_2 &= \norm{(\deg-\lap) - (\deg_\varepsilon - \lap_\varepsilon)}_2\notag\\
    &= \norm{(\deg - \deg_\varepsilon) - (\lap - \lap_\varepsilon)}_2 \notag\\
    &\leq  \norm{\deg - \deg_\varepsilon}_2 + \norm{\lap -\lap_\varepsilon}_2.
    \label{eq:bound-on-adjacency-matrix-diff}
\end{align}
% \begin{align}
%         \norm{\GEP_{\varepsilon}\v_\F-\GEP\v_\F }_2 &\leq 2\varepsilon \sqrt{N-S}\left(1+\frac{S}{N-S}\right),\\
%         \norm{\GEP\v_\F - \ones{N}}_2 &\leq \frac{\lambda_2(\GEP)}{d} \sqrt{\frac{NS}{N-S}}.
% \end{align}

According to the construction of $\GEP$ (cf. Sec.\ref{sec:sparse-graph-gradient-code-encoding-matrix-construction}), we apply the following Lemma on the right-hand-side terms of Eq.\eqref{eq:bound-on-adjacency-matrix-diff}.

\begin{lemma}
% [Modified{\cite[Theorem 3.3]{degree-preserving-sparsifier}}]
\label{lemma:degree-preserving-sparsifier}
Let $\GEP$ be the encoding matrix of an $(N,d,\varepsilon)$ expansion-preserving gradient code. If the degree parameter obeys $|d| \le N^\alpha$ for some fixed constant $\alpha>0$, then all entries of $\GEP$ satisfy 
$\GEP_{ij} \le N^{O(1)}$. 
Let $G$ be the graph whose biadjacency matrix is $\GEP$, and let $G_\varepsilon=\textsc{DegreePreservingSparsify}(G,\varepsilon)$ for $ \varepsilon\in(0,1]$.
Denote by $\lap$ and $\lap_\varepsilon$ the Laplacians of $G$ and $G_\varepsilon$, respectively.  
Since $G_\varepsilon$ is a degree-preserving $\varepsilon$-sparsifier of $G$, their Laplacians satisfy
$    \norm{\lap - \lap_\varepsilon}_2 \le (e^\varepsilon-1) \norm{\lap}_2.$
\end{lemma} 
The proof of \ref{lemma:degree-preserving-sparsifier} is provided in Sec.\ref{sec:proof-of-lemma-degree-preserving-sparsifier}.

By Lemma \ref{lemma:degree-preserving-sparsifier}, $G_\varepsilon$ is a degree-preserving $\varepsilon$-sparsifier of $G$, thus Eq.\eqref{eq:degree-preserved} and $\norm{\lap - \lap_\varepsilon}_2~\le~(e^\varepsilon-1)\norm{\lap}_2$ hold. So, We have
$
    \norm{\adj - \adj_\varepsilon}_2 \leq \norm{\lap - \lap_\varepsilon}_2 \le (e^\varepsilon-1) \norm{\lap}_2.
$
Since $\norm{\lap}_2 = \norm{\deg - \adj}_2 = d-\lambda_n$, then $\norm{\lap}_2 \le d-(-d) = 2d$.
Therefore, Eq.\eqref{eq:same-matrix-norm} can be bounded as
\begin{align}
%\label{eq:same-matrix-norm-upperbound}
    \norm{\GEP - \GEP_\varepsilon}_2 \leq 2(e^\varepsilon-1) d.
\end{align}
Combining this bound with Eqs.\eqref{eq:first-term} and \eqref{eq:decodingvec-l2norm}, it follows that
\begin{align*}
    \|\boldsymbol{\Delta}_1\|_2&=\|\GEP_\varepsilon\v_{\F} - \GEP\v_\F\|_2 \\
    &\leq 2(e^\varepsilon-1) \sqrt{N-S}\left(1+\tfrac{S}{N-S}\right),
\end{align*}
which gives the bound for $\boldsymbol{\Delta}_1$.

Now, we turn to bound $\boldsymbol{\Delta}_2$ in Eq.\eqref{eq:error-upperbound-decomposition}.
Given the following lemma, we may assume without loss of generality that $\u_1(\GEP) = \v_1(\GEP) = \ones{N}$ and 
\begin{align}
    \|\u_i(\GEP)\|_2 = \|\v_i(\GEP)\|_2 =1, i\in \{2, \dots, n\}.
    \label{eq:othornormal-singular-basis}
\end{align}
\begin{lemma}
\label{lem:biadj-matrix}
For the encoding matrix $\GEP_\varepsilon$ of an $(N, c, d, \varepsilon)$ expansion-preserving gradient code, we conclude that
\begin{enumerate}
\renewcommand{\labelenumi}{\roman{enumi})}
    \item $(\ones{N}, \ones{N}, d)\in \{(\u_i(\GEP_\varepsilon), \v_i(\GEP_\varepsilon), \sigma_i(\GEP_\varepsilon)\}_{i=1}^n$.
    \item $\sigma_2(\GEP_\varepsilon)< d$.
\end{enumerate}
\end{lemma}
The proof of Lemma \ref{lem:biadj-matrix} is deferred to Appendix \ref{proof of lem:biadj-matrix}.

We replace $\v_\F$ with Eq.\eqref{eq:decoding_vec} and obtain
\begin{align}
    \|\boldsymbol{\Delta}_2\|_2 &=\|\GEP \v_\F - \ones{N}\|_2\notag\\
    &= \|\tfrac{1}{d}\GEP(\ones{N}+ \a_\F) - \ones{N}\|_2
    = \|\tfrac{1}{d} \GEP \a_\F\|_2.
    \label{eq:l2norm_of_decoding}
\end{align}

% \begin{corollary}
% \label{coro:l2norm_of_aF}
%     For any non-empty set $\F\subset[N]$, there exists $\alpha_2, \dots, \alpha_N\in \RR$ such that $\a_\F = \alpha_2\v_2 + \dots \alpha_N \v_N(\GEP)$, and $\norm{\a_\F}_2 = \sqrt{\sum_{i=2}^N}\alpha_i^2 = \sqrt{\frac{NS}{N-S}}$.
% \end{corollary}

Due to the particular choice of $\a_\F$ (see Eq.\eqref{eq:decoding_vec}), we have the following lemma.
\begin{lemma}
\label{lem:l2norm_of_aF}
    For any non-empty set $\F\subset[N]$, there exists $\alpha_2, \dots, \alpha_N\in \RR$ such that $\a_\F= \alpha_2\v_2(\GEP) + \dots+ \alpha_N \v_N(\GEP)$, and $\norm{\a_\F}_2 = \sqrt{\textstyle\sum_{i=2}^N}\alpha_i^2 = \sqrt{\tfrac{NS}{N-S}}$.
\end{lemma}
The proof of Lemma \ref{lem:l2norm_of_aF} is provided in Appendix \ref{proof of lem:l2norm_of_aF}.
Therefore, Eq.\eqref{eq:l2norm_of_decoding} can be rewritten as
\begin{align}
\|\tfrac{1}{d} \GEP \a_\F\|_2 &= \|\tfrac{1}{d}\GEP \left(\textstyle\sum_{i=2}^N \alpha_i\v_i(\GEP)\right)\|_2 \notag\\
    &\overset{(d)}{=} \|\tfrac{1}{d} \GEP \textstyle\sum_{i=2}^N \alpha_i \sigma_i(\GEP) \v_i(\GEP)\|_2\notag\\
    &=\tfrac{1}{d} \sqrt{\textstyle\sum_{i=2}^N \alpha_i^2\sigma_i^2(\GEP)}\overset{(e)}{\leq} \tfrac{\sigma_2(\GEP)}{d} \sqrt{\tfrac{NS}{N-S}}\notag,
\end{align}
where (d) is because of Eq.\eqref{eq:othornormal-singular-basis} and (e) follows from Lemma~\ref{lem:l2norm_of_aF}, which finishes the bound for $\norm{\boldsymbol{\Delta}_2}_2$. 
\end{proof}

\subsection{Proof of Lemma \ref{lem:norm-of-matrices-with-standard-gaussian-entries}}
We begin with introducing sub-Gaussian random variables and sub-Gaussian norm.
\begin{definition}[{\cite[Definition 2.5.6 \& Proposition 2.5.2]{vershynin2020high}}]
\label{def:sub-gaussian}
    The sub-Gaussian norm of a random variable $X$ is defined as
    \[
        \psitwo{X} = \inf\{ t > 0 : \ex{\exp(X^2 / t^2)} \leq 2 \}.
    \]
    If $\psitwo{X}$ is finite, we say $X$ is sub-Gaussian.
\end{definition}
Then, we prove Lemma \ref{lem:norm-of-matrices-with-standard-gaussian-entries} as follows.

\label{proof of lem:norm-of-matrices-with-standard-gaussian-entries}
\begin{proof}
Let $\mu \triangleq \mathbb E|G|=\sqrt{\frac{2}{\pi}},$
where $G\sim\gaussian{0}{1}$. Define a symmetric random matrix $\boldsymbol W$ by
\[
\GEP_0= (\mu+c)\mathbf 1\mathbf 1^\top + \boldsymbol W,
\]
where $\ones{}\in \RR^{N-1}$. Equivalently, for $1\le i\le j\le N-1$,
$\boldsymbol{W}_{ij}=|G_{ij}|-\mu,$
and $\boldsymbol{W}_{ji}=\boldsymbol{W}_{ij}$ for $i>j$.

Here are properties of $\boldsymbol W$. Since the random variables $\{G_{ij}:1\le i\le j\le N-1\}$ are i.i.d., the collection $\{\boldsymbol{W}_{ij}:1\le i\le j\le N-1\}$
is independent. Moreover,
$\mathbb E[\boldsymbol{W}_{ij}]
=
\mathbb E|G_{ij}|-\mu
=
0.$
Note that $|G|$ is sub-Gaussian because $\Pr(|G|>t)\le 2e^{-t^2/2}\text{ for } t\ge 0$.
\begin{lemma}[{\cite[Lemma 2.6.8]{vershynin2020high}}]
\label{lem:centering-sug-gaussian}
    If $X$ is a sub-Gaussian random variable then $X - \ex{X}$ is a sub-gaussian, too, and
    \[
        X -\ex{X} \leq C\psitwo{X},
    \]
    where $C$ is an absolute constant.
\end{lemma}
Hence, by Lemma~\ref{lem:centering-sug-gaussian}, $\boldsymbol{W}_{ij}$ is i.i.d. sub-Gaussian with mean zero and
% with $\psi_2$-norm bounded by an absolute constant
there exists an absolute constant $K>0$ such that
$\|W_{ij}\|_{\psi_2}\le K\text{, for all }1\le i\le j\le N-1.$
Now, we apply the following Lemma on $\boldsymbol{W}$.
\begin{lemma}[Modified {\cite[Corollary 4.4.8]{vershynin2020high}}]
    Let $\boldsymbol{A}$ be an $n\times n$ symmetric random matrix whose entries $\boldsymbol{A}_{ij}$ on and above the diagonal are i.i.d., mean zero, sub-gaussian random variables. Then, for any $t> 0$ we have
    \[
    \prob{\|\boldsymbol{A}\|\geq CK(\sqrt{n} + t)}< 4\exp(-t^2),
    \]
    where $C$ is an absolute constant and $K = \psitwo{\boldsymbol{A}_{ij}}$.
\end{lemma}
So, there exist absolute constants $C_1>0$ such that for any $t\ge 0$, 
$\prob{\|\boldsymbol W\|_2 \leq C_1(\sqrt{N-1} + t)}
\geq 1- 4\exp(-t^2).$ 
% Set $t=\varepsilon\sqrt n.$ Then
% $\prob{\|\boldsymbol W\|_2 \ge C_1\sqrt{n}(1+\delta)}
% \le 4\exp(-\delta^2 n).$

On the other hand, since $\mathbf 1\mathbf 1^\top$ is rank one, $\|\mathbf 1\mathbf 1^\top\|_2=\|\mathbf 1\|_2^2=N-1.$ Therefore, $\|(\mu+c)\mathbf 1\mathbf 1^\top\|_2 = |\mu+c|(N-1).$ 
% If $c\ge -\mu$, then $\mu+c\ge 0$, and so
% $\|(\mu+c)\mathbf 1\mathbf 1^\top\|_2 = (\mu+c)(N-1)$. 
Hence, by the triangle inequality,
\[
\|\GEP_0\|_2
\le
\|(\mu+c)\mathbf 1\mathbf 1^\top\|_2 + \|\boldsymbol W\|_2
=
|\mu+c|(N-1)+\|\boldsymbol W\|_2.
\]
Combining this with the above tail bound for $\|\boldsymbol W\|_2$, we obtain
\[
\mathbb{P}[
\|\GEP_0\|_2
\geq
|\mu+c|(N-1) + C_1(\sqrt{N-1}+t)
]
< 4e^{-t^2}.\qedhere
\]
\end{proof}

\subsection{Proof of Lemma \ref{lemma:psi-two-upper-bound}}
\begin{proof}
\label{proof of lemma:psi-two-upper-bound}
    By Definition \ref{def:sub-gaussian}, in order to obtain $\psitwo{X_i}$, we need to calculate $\mathbb{E}\big[{\exp{\left(X_i^2/C^2\right)\big]}}$ first. Plug $X_i =\rho \sum_{j=1}^{N-S} G_{ij}B_{ij}-1$, where $G_{ij} \overset{\text{i.i.d.}}{\sim}~\gaussian{a}{b}$ and $B_{ij}\overset{\text{i.i.d.}}{\sim}\text{Bern}(\gamma)$, into $\mathbb{E}\big[{\exp{\left(X_i^2/C^2\right)\big]}}$, we have
    \begin{equation*}
        \ex{\exp{\left(\tfrac{X_i^2}{C^2}\right)}} =\ex{\exp{\left[\left(\textstyle{\rho\sum_{j=1}^{N-S}G_{ij}B_{ij}-1}\right)^2/C^2\right]}}.
    \end{equation*}
    We use the law of total expectation and obtain
    \begin{align*}&\ex{\exp{\bigg(\Big(\textstyle{\rho\sum\limits_{j=1}^{N-S}G_{ij}B_{ij}-1}\Big)^2/C^2\bigg)}}\\
        &=\ex{\ex{e^{\left({\rho\sum_{j=1}^{N-S}G_{ij}B_{ij}-1}\right)^2/C^2}\Big|\{B_{ij}:j \in [N-S]\}}}\\
        &\overset{(a)}{=} \textstyle\sum\limits_{p=0}^{N-S}\tbinom{N-S}{p}\gamma^p(1-\gamma)^{N-S-p}\mathbb{E}\Big[{e^{\tfrac{\big(\rho\sum_{j=1}^{p}G_{ij}-1\big)^2}{C^2}}}\Big]\\
        &
        = \textstyle\sum\limits_{p=0}^{N-S}\tbinom{N-S}{p}\gamma^p(1-\gamma)^{N-S-p} \mathbb{E}\Big[e^{\tfrac{\big(\rho\sum_{j=1}^{p}G_{ij}-\rho pa+\rho pa-1\big)^2}{C^2}}\Big]\\
        &
        = \textstyle\sum\limits_{p=0}^{N-S}\tbinom{N-S}{p}\gamma^p(1-\gamma)^{N-S-p} \mathbb{E}\Big[\exp\Big({\tfrac{\big(Y+\beta\big)^2}{C^2}}\Big)\Big],
    \end{align*}
    where $(a)$ comes from $B_{ij}\overset{\text{i.i.d.}}{\sim}\text{Bern}(\gamma)$. Let $Y \triangleq \rho\sum\limits_{j=1}^{p}G_{ij}-\rho pa$ and $\beta \triangleq \rho pa-1$. 
    % Since $G_{ij} \sim~\gaussian{a}{b}$, then
    Therefore, $Y\sim \gaussian{0}{\sigma^2}$ where $\sigma^2 = \rho^2(pb+(p^2-p)c)$.

\begin{lemma}
% [Expectation of exponential of a parabola of a Gaussian random variable]
    \label{lemma:expectation-of-exponential-of-the-linear-of-a-Gaussian-random-variable}
If a random variable $G\sim \gaussian{\mu}{\sigma^2}$, then for $\theta \in \RR$ and $\alpha\leq \frac{1}{2\sigma^2}$,
\begin{align*}
    &\ex{\exp{(\alpha G^2 + \theta G)}}= \tfrac{1}{\sqrt{\frac{1}{\sigma^2} - 2\alpha}}\exp{\Big(- \tfrac{\sigma^2(\theta+2\mu)^2}{4\alpha \sigma^2-2} - \tfrac{\mu^2}{2\sigma^2}\Big)}.
\end{align*}
\end{lemma}
The proof of Lemma \ref{lemma:expectation-of-exponential-of-the-linear-of-a-Gaussian-random-variable} is deferred to Appendix \ref{proof of lemma:expectation-of-exponential-of-the-linear-of-a-Gaussian-random-variable}.

    We apply Lemma \ref{lemma:expectation-of-exponential-of-the-linear-of-a-Gaussian-random-variable} and have
    \begin{align}\label{eq:expectation-of-exp}  &\mathbb{E}\bigg[\exp{\Big(\tfrac{1}{C^2}\big(\rho\textstyle\sum_{j=1}^{p}G_{ij}-\rho pa+\rho pa-1\big)^2\Big)}\bigg] \notag\\ 
    & \quad=\mathbb{E}\Big[\exp{\big(\tfrac{Y^2+2\beta Y+\beta^2}{C^2}\big)}\Big]\notag\\
    &\quad =\tfrac{1}{\sqrt{\tfrac{1}{\sigma^2}-\tfrac{2}{C^2}}}\exp{\Big(-\tfrac{4\sigma^2 (\tfrac{\beta}{C^2})^2}{\tfrac{4\sigma^2}{C^2}-2} + \tfrac{\beta^2}{C^2} \Big)}\notag\\
     &\quad=\tfrac{1}{\sqrt{\tfrac{1}{\sigma^2}-\tfrac{2}{C^2}}}\exp{\Big(
    \tfrac{\beta^2}{C^2}\big(1-\tfrac{2\sigma^2}{2\sigma^2-C^2}\big)
     \Big)}.
    \end{align}%
    Note that one essential condition for $C$ is 
    % \begin{equation*}
        $2\sigma^2 <C^2,$
    % \end{equation*}
     which follows from $\frac{1}{\sigma^2}-\frac{2}{C^2} >0$. 
    Let $\Cpsitwo$ be an upper bound of $\psitwo{X_i}$, i.e., 
    % \begin{equation*}
    $\Cpsitwo \geq \psitwo{X_i}.$
    % \end{equation*}
    By Definition \ref{def:sub-gaussian}, it follows that
         % \label{cond-1}
        $2\sigma^2<\Cpsitwo^2.$
     We assume that $C^2 = q \sigma^2$ where $q\in \RR$. Therefore, the right-hand-side of Eq. \eqref{eq:expectation-of-exp} can be simplified as
   \begin{align}
    % &\tfrac{1}{\sqrt{\tfrac{1}{\sigma^2}-\tfrac{2}{C^2}}}\exp{\left(
    % \tfrac{\beta^2}{C^2}\left(1-\tfrac{2\sigma^2}{2\sigma^2-C^2}\right)
    %  \right)}
     \label{eq:simplified-expectation-of-exp}
     % \\
     %    &=\tfrac{1}{\sqrt{\tfrac{1}{\sigma^2}(1-\tfrac{2}{q})}}
     %    \exp{\left({
     %        \tfrac{\beta^2}{\sigma^2}\left(\tfrac{1}{q}-\tfrac{2}{2q-q^2}\right)}
     %    \right)}\notag
        % &=\tfrac{\sigma}{\sqrt{(1-\tfrac{2}{q})}}
        % \exp{\left(
        %     \tfrac{\beta^2}{\sigma^2}\left(\tfrac{1}{q}-\tfrac{2}{q(2-q)}\right)
        % \right)}
        % &\quad=\tfrac{\sigma}{\sqrt{(1-\tfrac{2}{q})}}
        % \exp{\left(
        %     \tfrac{\beta^2}{\sigma^2}\left(\tfrac{2-q-2}{q(2-q)}\right)
        % \right)}
        % =
        \tfrac{\sigma}{\sqrt{(1-\tfrac{2}{q})}}
        \exp{\Big(
            \tfrac{\beta^2}{\sigma^2}\big(\tfrac{1}{q-2}\big)
        \Big)}.
    \end{align}
    Since $1-\tfrac{2}{q}>0$ is required, we must have $q>2.$
    % \begin{equation}
    % \label{ineq:condition-of-q}
    %     q>2.
    % \end{equation}
    To avoid 
    % violating Ineq. \eqref{ineq:condition-of-q}, 
    this constraint,
    we can assume $q\geq 4$ which implies that $\sqrt{1-2/q}\geq \sqrt{1/2}$. Then, Eq. \eqref{eq:simplified-expectation-of-exp} can be upper bounded by $\sqrt{2}\sigma \exp{\left(
        \frac{\beta^2}{\sigma^2}\left(\frac{1}{q-2}
        \right)\right)}.$
    Now, we further assume that $\sqrt{2}\sigma \exp{\left(
            \frac{\beta^2}{\sigma^2}\left(\frac{1}{q-2}\right)
        \right)}\overset{\mathrm{set}}{\leq} 
        2$ and obtain
     \begin{align}
        % \exp{\left(
        %     \tfrac{\beta^2}{\sigma^2}
        %     \left(
        %         \tfrac{1}{q-2}
        %     \right)
        % \right)} 
        % &\leq 
        % \tfrac{\sqrt{2}}{\sigma}\notag\\
        \tfrac{\beta^2}{\sigma^2}
        \left(
            \tfrac{1}{q-2}
        \right) 
        \leq 
        \ln{(\tfrac{\sqrt{2}}{\sigma})}\label{eqn:sigma_bound}
    \end{align}
    We must have $\sigma < \sqrt{2}$ to be satisfied since $\frac{\beta^2}{\sigma^2}(\frac{1}{q-2})>0$ when $q\ge4$. Therefore, Inequality \eqref{eqn:sigma_bound} can be rewritten as 
        \begin{equation}
        \label{eqn:q_bound}
            2+\tfrac{\beta^2}{\sigma^2} \tfrac{1}{\ln{(\sqrt{2}/\sigma)}} \leq q,
        \end{equation}
    Therefore, conditioned on $C^2 = q\sigma^2$ with $q\geq4$, $\sigma<\sqrt{2}$, and Eq. \eqref{eqn:q_bound}, 
    % we obtain
     $C^2$ can be bounded as follows:
    \begin{align*}
        C^2 &\ge \max \left\{ \left(2+\tfrac{\beta^2}{\sigma^2}\tfrac{1}{\ln{(\sqrt{2}/\sigma)}}\right)\sigma^2, 4\sigma^2\right\}\\
        &= \max\left\{2\sigma^2 + \tfrac{\beta^2}{\ln{(\sqrt{2}/\sigma)}}, 4\sigma^2 \right\}.
    \end{align*}
    Recall that $\sigma^2 = \rho^2(pb+(p^2-p)c)$ and $\beta = \rho pa-1$.
    Let $f(p) = \rho^2(p^2c+p(b-c)) = \rho^2p(pc+(b-c))$ and $g(p)=\rho p a-1$, for $p = 0,1,\ldots,N-S$.
  Since the covariance matrix $\Sigmaa$ is a positive semi-definite matrix, $b-c\geq 0$ is required. So, if $c>0$, $f(p)$ increases with the increment of $p$; if $c=0$, then $f(p)=0$ for $p = 0,1,\ldots,N-S$. Furthermore, $g(p) = \rho pa-1$ doesn't decrease with the increment of $p$ if $\rho a\geq 0$. So, we claim that if $c\geq 0$ and $\rho a\geq 0$, then $2\sigma^2 + \frac{\beta^2}{\ln{(\sqrt{2}/\sigma)}}$ doesn't decrease with the increment of $p$. Denote $f(N-S)$ and $g(N-S)$ as $\sigmamax$ and $\betamax$ respectively, then
     \begin{equation}
         {\Cpsitwo}^2 = \max\left\{2\sigmamax^2 + \tfrac{\betamax^2}{\ln{(\frac{\sqrt{2}}{\sigmamax})}}, 4\sigmamax^2 \right\},
     \end{equation}
     where ${\sigmamax}^2 = \rho^2((N-S)b+((N-S)^2-(N-S))c)$ and $\betamax = \rho(N-S)a-1$.
    By Definition \ref{def:sub-gaussian} and the assumption of $\Cpsitwo$ being an upper bound of the sub-Gaussian norm of $X_i$, $\psitwo{X_i}<\infty$ can be obtained. So, Lemma \ref{lemma:psi-two-upper-bound} follows.
\end{proof}

\subsection{Proof of Lemma \ref{lemma:centering-a-sub-exponential-random-variable}}
\begin{proof}
\label{proof of lemma:centering-a-sub-exponential-random-variable}
By the triangle inequality of a norm, we have
% $\psione{\cdot}$ is a norm, we can use triangle inequality and get
\begin{align}
\label{eq:bound-for-centered-sub-exponential-rv}
    \psione{X-\ex{X}}\leq \psione{X} - \psione{\ex{X}}.
\end{align}
Thus it suffices to bound the second term. According to definition \ref{def:sub-exponential}, for a random variable $Z \triangleq \ex{X}$, we trivially have $\psione{\ex{X}}\lesssim|\ex{X}|$. By Jensen’s inequality, we then have
% \begin{align*}
     $|\ex{X}|\leq \ex{|X|} = \norm{X}_1.$
% \end{align*}
\begin{lemma}[Modified {\cite[proposition 2.7.1]{vershynin2020high}}] Let $X$ be a sub-exponential random variable, the moments of $X$ satisfy that for $K>0$,
\[
    \norm{X}_{L^{p}} = \ex{|X|^{p}} ^{1/p}\leq Kp,\quad\text{for all }p\geq1.
\]
\end{lemma}
Apply this lemma and let $p=1$, we obtain $\norm{X}_1\lesssim \psione{X}.$
% \begin{align}
%     \norm{X}_1\lesssim \psione{X}.
% \end{align}
Substitute this into Eq.\eqref{eq:bound-for-centered-sub-exponential-rv}, we complete the proof.
\end{proof}

\subsection{Proof of Lemma \ref{lemma:expectation-of-two-norm-of-X}}
\begin{proof}
% [Proof of Lemma~\ref{lemma:expectation-of-two-norm-of-X}]
\label{proof of lemma:expectation-of-two-norm-of-X}
Assume an $(N, K, L, R, \lambda, \gamma)$-SG gradient code exists, and let $\GSG \in \RR^{K\times N}$ denote its encoding matrix. 
Let $\F \subset [N]$ denote the non-straggler set, with $|\F| = N-S$. 
% Define the decoding vector $\vopt(\F) \in \RR^{N}$ such that 
% \[
% \v_i(\F) =
% \begin{cases}
% \rho, & \text{if } i \in \F,\\
% 0, & \text{otherwise},
% \end{cases}
% \]
% where $\rho$ is the scalar decoding coefficient. 
Define the decoding vector $\vopt(\F) \in \RR^{N}$ such that, for each $i \in [N]$, the $i$-th entry satisfies $\v_i(\F) = \rho$ if $i \in \F$, and $\v_i(\F) = 0$ otherwise, where $\rho$ is the scalar decoding coefficient.
Then for each row $i \in [K]$, define
\[
X_i = (\GSG\v(\F))_i - 1 = \rho \textstyle\sum_{j\in\F}\GSG_{ij} - 1.
\]
The expectation of the sum of squared deviations is
\begin{align}
&\ex{\textstyle\sum_{i=1}^{K}X_i^2} \notag = \mathbb{E}[\|\GSG\v(\F) - \ones{K}\|_2^2]\notag\\
& = \mathbb{E}[\ones{K}^\top\ones{K} - 2\rho\ones{N}^\top \ones{\F}^\top {\GSG}^\top \ones{K} \rho^2\, \ones{\F}^\top {\GSG}^\top \GSG \ones{\F}]\notag\\
& = \ones{K}^\top\ones{K}
    - 2\rho\,\ones{N}^\top\mathbb{E}[{\GSG}^\top\ones{K}]\ones{\F}\notag+ \rho^2\, \ones{\F}^\top \mathbb{E}[{\GSG}^\top\GSG] \ones{\F},
    \label{eq:expectation-of-the-sum-of-squared-deviations}
\end{align}
where $\ones{\F}\in\{0,1\}^N$ is the indicator vector of $\F$ (i.e., $[\ones{\F}]_i = 1$ if $i\in\F$, $0$ otherwise). 
% Since the expectation is taken over the randomness of $\GSG$, we have
% \begin{align*}
% \ex{\sum_{i=1}^{K}X_i^2}
% &= \ones{K}^\top\ones{K}
%     - 2\rho\,\ones{N}^\top\ex{\GSG^\top\ones{K}}\!\odot\!\mathbbm{1}_{\F}
%     + \rho^2\, \mathbbm{1}_{\F}^\top \ex{\GSG^\top\GSG} \mathbbm{1}_{\F},
% \end{align*}
% where $\odot$ denotes elementwise multiplication. 
Because $\GSG$ is row-wise i.i.d. Gaussian with expected value satisfying Eq.~\eqref{row-exp}, we obtain $\ex{{\GSG}^\top\ones{K}} = L\ones{N}.$
% \[
% \ex{{\GSG}^\top\ones{K}} = L\ones{N}.
% \]
Hence,
\[
\ones{N}^\top\ex{{\GSG}^\top\ones{K}}\ones{\F} = L(N-S).
\]

Next, note that $\ex{{\GSG}^\top\GSG}$ is an $N\times N$ matrix with diagonal entries $\ex{\boldsymbol{e}_i^\top\boldsymbol{e}_i}=L$ and off-diagonal entries $\ex{\boldsymbol{e}_i^\top\boldsymbol{e}_j}=\lambda$ for $i\neq j$, where $\boldsymbol{e}_i$ denotes the $i$-th column of $\GSG$. 
Thus, $\ex{{\GSG}^\top\GSG} = (L-\lambda)\ident{N} + \lambda\allones{N\times N}.$
% \[
% \ex{{\GSG}^\top\GSG} = (L-\lambda)\ident{N} + \lambda\allones{N\times N}.
% \]
Therefore,
\begin{align*}
\ones{\F}^\top\ex{{\GSG}^\top\GSG}\ones{\F}&= (L-\lambda)\ones{\F}^\top\ones{\F}
    + \lambda(\ones{\F}^\top\ones{N})(\ones{N}^\top\ones{\F})\\
&= (L-\lambda)(N-S) + \lambda(N-S)^2.
\end{align*}

% Substituting these into the expectation expression gives
% \begin{align*}
% &\ex{\textstyle\sum_{i=1}^{K}X_i^2}\\
% &= K - 2\rho L (N-S)+ \rho^2\big((L-\lambda)(N-S) + \lambda(N-S)^2\big).
% \end{align*}
Hence, $\ex{\textstyle\sum_{i=1}^{K}X_i^2}$ equals the following expression
\begin{align*}
K - 2\rho L (N-S)+ \rho^2\big((L-\lambda)(N-S) + \lambda(N-S)^2\big).
\end{align*}

Setting $\rho = \frac{L}{L+\lambda(N-S-1)}$, this expectation yields
\begin{align*}
&\ex{\textstyle\sum_{i=1}^{K}X_i^2}
% &\quad= K - \tfrac{2L^2(N-S)}{L+\lambda(N-S-1)}+ \tfrac{L^2\big((L-\lambda)(N-S) + \lambda(N-S)^2\big)}{(L+\lambda(N-S-1))^2}\\
= K\left(1 - \tfrac{1}{K}\tfrac{L^2(N-S)}{L + \lambda (N-S-1)}\right) \quad
\end{align*}
By Eq.\eqref{eq:bibd-error}, we have $\mathbb{E}[\textstyle\sum_{i=1}^{K}X_i^2] = K\cdot\err{\GB},$
% \begin{align*}
%     \ex{\textstyle\sum_{i=1}^{K}X_i^2} = K\cdot\err{\GB},
% \end{align*}
where $\GB$ is the encoding matrix of an $(N, K, L, R, \lambda)$ BIBD-GC.
\end{proof}

% \subsection{Proof of Lemma \ref{rmk:parameter-selection}}
% % TODO: complete this proof.
% \begin{proof}
%     \label{proof of rmk:parameter-selection}
% \end{proof}

\subsection{Proof of Lemma \ref{lemma:degree-preserving-sparsifier}}
\label{sec:proof-of-lemma-degree-preserving-sparsifier}
In this proof, we retain the notation introduced in Sec. \ref{sec:sparse-graph-gradient-code-encoding-matrix-construction}. We begin by recalling a key guarantee of the algorithm \textsc{DegreePreservingSparsify}.

\begin{lemma}[modified from {\cite[Theorem 3.3]{degree-preserving-sparsifier}}]
For any $\varepsilon \in (0,1]$, every undirected graph $G$ whose edge weights are
(i) strictly positive and
(ii) polynomially bounded, meaning that each weight satisfies $w_{ij} \le N^{c}$ for some absolute constant $c > 0$,
admits a degree-preserving $\varepsilon$-sparsifier. The algorithm \textsc{DegreePreservingSparsify} takes $G$ as input and returns such a degree-preserving $\varepsilon$-sparsifier $G_\varepsilon$.
\end{lemma}

Before applying the algorithm, we verify that the weighted adjacency matrix $\GEP$ of the resulting sparsifier satisfies the assumptions required for our analysis:
\begin{enumerate}
\item each entry $\GEP_{ij}$ is positive;
\label{condition1}
\item each entry is polynomially bounded, i.e., \(\GEP_{ij} \le N^{C}\) for some constant $C > 0$.
\label{condition2}
\end{enumerate}

% Let $\boldsymbol{\Delta} = \lap_\varepsilon - \lap$.
Once Conditions \ref{condition1}–\ref{condition2} hold, the returned graph $G_\varepsilon$ is indeed a degree-preserving $\varepsilon$-sparsifier of $G$, enabling us to bound the operator norm of the difference between their Laplacian matrices.
\begin{proof}
We first validate the condition \ref{condition1}. Since the entries of $\GEP_0$ are i.i.d. drawn from standard half-normal distribution, the entries are positive. According to the constraints of $d$ summarized in Sec.\ref{sec:sparse-graph-feasible-parameter-regime}, $d> \max\limits_{1\leq i\leq N-1}\sum_{j}^{N-1}(\GEP_0)_{ij}$, which implies that $\boldsymbol{m}_i = d-\sum_{j}^{N-1}(\GEP_0)_{ij}>0$. Recall that
\begin{align*}
    \alpha = (2-N)d + \textstyle\sum_{i=1}^{N-1}\sum_{j=1}^{N-1} (\GEP_0)_{ij},
\end{align*}
where every term is positive, and thus $\alpha>0$. Therefore, the entries of $\GEP_{ij}>0$ for all $1\leq i,j\leq N$.

We then validate the condition \ref{condition2}.
According to the construction described in Sec.\ref{sec:sparse-graph-gradient-code-encoding-matrix-construction}, $(\GEP_0)_{ij}$ follows i.i.d. standard half-normal distribution, that is, for $1\leq i, j\leq N-1$, $\GEP_{ij}$ follows i.i.d. standard half-normal distribution. Therefore, we obtain that, for $1\leq i, j\leq N-1$ and $t_1\geq 0$,
\begin{align}
    \prob{\GEP_{ij}\geq t_1}  &= \prob{(|\gaussian{0}{1}|\geq t_1)} \notag\\
    &= \prob{\gaussian{0}{1}\geq 1} + \prob{\gaussian{0}{1}\leq -1}\notag\\
    &\overset{(a)}{=} 2\prob{\gaussian{0}{1}\geq 1}\notag\\
    &\overset{(b)}{\leq} 2\exp(-t_1^2/2),\notag
\end{align}
where (a) is because of the symmetry of standard Gaussian distribution and (b) follows by applying Chernoff bound on a standard Gaussian random variable.

Let $M_0 = \max_{1\leq i, j\leq N-1}\GEP_{ij}$. By union bound, we have
\begin{align}
\label{eq:max-E_0-tail-bound}
    &\prob{M_0\geq t_1} = \prob{\cup^{(N-1)}_{i, j=1}\{\GEP_{ij}\geq t_2\}}\notag\\
    &\leq \textstyle\sum_{i}^{N-1} \sum_{j}^{N-1} \prob{\GEP_{ij}\geq t_2}\notag\\
    &\leq 2(N-1)^2\exp (-t_2^2/2).
\end{align}

Choose $t_2 = \sqrt{2(1+\delta)\ln(N-1)^2}$, the right-hand-side of Eq.\eqref{eq:max-E_0-tail-bound} can be simplified to $2(N-1)^{2\delta}$.
% \begin{align*}
%     &2(N-1)^2\exp (-t^2/2) \\
%     &\quad= 2(N-1)^2 \exp(-(1+\delta)\log(N-1)^2)\\
%     &\quad= 2(N-1)^2(N-1)^{-2(1+\delta)} = 2(N-1)^{2\delta}.
% \end{align*}
Therefore, with probability $1-2(N-1)^{2\delta}$,
\begin{align*}
    M_0\leq \sqrt{2(1+\delta)\ln(N-1)^2} = \sqrt{4(1+\delta)\ln(N-1)}.
\end{align*}
Since $\sqrt{\ln(N)}  = o(N^{C_1})$ for some fixed $C_1>0$ and $1\leq i, j\leq N-1$, we have
\begin{align}
\label{eq:scheduling-matrix-entries-poly-bounded}
    \GEP_{ij} \leq O(\sqrt{\ln(N)}) = o(N^{C_1}) \leq O(N^{C_1}).
\end{align}
Then we analyze entries of $\boldsymbol{m}$ and $\alpha$, i.e. $\{\GEP_{ij}\}_{i=N}\cup \{\GEP_{ij}\}_{j=N}$. Let $S_i:= \sum_{j=1}^{N-1}(\GEP_0)_{ij}$ for $1\leq i\leq N-1$. So, $\boldsymbol{m}_i = d-S_i$ for $1\leq i\leq N-1$.
Because $(\GEP_0)_{ij}$ follow standard half-normal distribution i.i.d. for $1\leq i\leq j\leq N-1$. Let $\mu = \ex{|\gaussian{0}{1}|} = \sqrt{2/\pi}$ and thus $\ex{S_i} = \mu(N-1)$ by linearity of expectations.

Since a standard half-normal distribution has the same tail as a standard Gaussian distribution, then a standard half-normal distribution is also a sub-Gaussian random variable. According to Definition \ref{def:sub-gaussian}, for some constant $C_2>0$, we have
\begin{align*}
    \prob{|S_i - \mu(N-1)|\geq t_3} \leq 2\exp(-C_2 t_3^2/N).
\end{align*}
Choose $t_3 = C_3\sqrt{N\ln N}$ with $C_3>0$, then
\begin{align*}
    \prob{|S_i - \mu(N-1)|\geq C_3\sqrt{N\ln N}}  
    % &\leq 2\exp(-C_2C_3^2\ln N) \\
    &\leq 2N^{-C_2C_3^2}.
\end{align*}
Take union bound over all $i\leq N-1$, then we have
\begin{align*}
    \mathbb{P}\Big[\max_{1\leq i\leq N-1} |S_i - \mu(N-1)|\geq C_3\sqrt{N\ln N}\Big]&\leq 2N^{1-C_2C_3^2}.
\end{align*}
By triangle inequality, we have
\begin{align*}
    \max_{1\leq i\leq N-1} |S_i| + |\mu(N-1)| \geq \max_{1\leq i\leq N-1} |S_i - \mu(N-1)|.
\end{align*}
Thus, for constant $C_2, C_3 >0$, we have
\begin{align*}
    \max_{1\leq i\leq N-1} |S_i| \leq C_3 \sqrt{N\ln N} - \mu(N-1) = O(N),
\end{align*}
with probability of $1-2N^{1-C_2C_3^2}$. When $C_2C_3^2$ is large enough, $1-2N^{1-C_2C_3^2}\rightarrow 0$.
With assumption of $|d|\leq N^{C}$ with some fixed $C>0$, we obtain that, for $1\leq i\leq N-1$,
\begin{align*}
    |\boldsymbol{m}_i| = |d-S_i| \leq |d|+|S_i| \overset{(c)}{\leq} N^{C} + O(N) \leq C_4N^{\max(C, 1)},
\end{align*}
where (c) follows from triangle inequality.
So every entry in $\boldsymbol{m}$ is bounded by $N^{C_4}$ for some constant $C_4 = \max(C, 1)$.

We now turn to analyze $\alpha$ as follows:
\begin{align*}
    \alpha &= (2-N)d + \textstyle\sum_{i=1}^{N-1}\sum_{j=1}^{N-1} (\GEP_0)_{ij}\\
    &\overset{(d)}{=}(2-N)N^{C} + \textstyle\sum_{i=1}^{N-1} O(N^{C_1})\\
    &= (2-N)N^{C} + O(N^{2C_1}) \leq C_5 N^{\max(2C_1, C+1)},
    % &\overset{(d)}{=} (2-N)N^{C} + O(N^{2C_1}) 
    % \leq C_5 N^{\max(2C_1, C+1)},
\end{align*}
where (d) comes from Eq.\eqref{eq:scheduling-matrix-entries-poly-bounded}.
Similarly, $\alpha$ is also bounded by $N^{C_6}$, for some constant $C_6 = \max(2C_1, C+1)$. Therefore, we complete the validation that the entries of $\GEP$ are polynomially bounded.

% Hence, $G_\varepsilon$ is a degree-preserving $\varepsilon$-sparsifier of $G$. So, for every $x \in \mathbb{R}^{2N}$, Eq.\eqref{eq:epsilon-sparisifer-quadratic-form} holds. Subtracting $\xx{}^\top \lap \xx{}$ from all sides of Eq.\eqref{eq:epsilon-sparisifer-quadratic-form} gives
% \[
%     -\varepsilon \xx{}^\top \lap \xx{} \le \xx{}^\top (\lap_\varepsilon - \lap) \xx{} \le \varepsilon \xx{}^\top \lap \xx{}.
% \]
% Equivalently,
% \[
%     |\xx{}^\top \boldsymbol{\Delta} \xx{}| \le \varepsilon \xx{}^\top \lap \xx{}.
% \]
% For any unit vector $\norm{\xx{}}_2 = 1$, we further have
% \[
%     |\xx{}^\top \boldsymbol{\Delta} \xx{}| \le \varepsilon\, \xx{}^\top \lap \xx{} 
%     \le \varepsilon\, \lambda_{1}(\lap) \norm{\xx{}}_2^2 = \varepsilon\, \norm{\lap}_2.
% \]
% Taking the supremum over all $\norm{\xx{}}_2 = 1$ and using the variational characterization of the spectral norm for symmetric matrices, we obtain
% \[
%     \norm{\boldsymbol{\Delta}}_2 = \sup_{\norm{\xx{}}_2 = 1} |\xx{}^\top \boldsymbol{\Delta} \xx{}| \le \varepsilon \norm{\lap}_2.
% \]
% Thus,
% \[
%     \norm{\lap - \lap_\varepsilon}_2 \le \varepsilon \norm{\lap}_2.
% \]

Assume that $G_\varepsilon$ is a degree-preserving $\varepsilon$-sparsifier of $G$ in the
multiplicative sense, namely, for all $\xx{}\in\mathbb{R}^{2N}$,
\begin{equation}
\label{eq:epsilon-sparsifier-multiplicative}
    e^{-\varepsilon}\,\xx{}^\top \lap \xx{}
    \le
    \xx{}^\top \lap_\varepsilon \xx{}
    \le
    e^{\varepsilon}\,\xx{}^\top \lap \xx{}.
\end{equation}
Subtracting $\xx{}^\top \lap \xx{}$ from all sides of
Eq.~\eqref{eq:epsilon-sparsifier-multiplicative} yields
\[
    \bigl(e^{-\varepsilon}-1\bigr)\,\xx{}^\top \lap \xx{}
    \le
    \xx{}^\top (\lap_\varepsilon-\lap) \xx{}
    \le
    \bigl(e^{\varepsilon}-1\bigr)\,\xx{}^\top \lap \xx{}.
\]
% Starting from the multiplicative sparsifier assumption, for all $\xx{}\in\mathbb{R}^{2N}$,
% \begin{equation*}
% \label{eq:lap-multiplicative-assumption}
%     e^{-\varepsilon}\,\xx{}^\top \lap \xx{}
%     \le
%     \xx{}^\top \lap_\varepsilon \xx{}
%     \le
%     e^{\varepsilon}\,\xx{}^\top \lap \xx{}.
% \end{equation*}
% Subtracting $\xx{}^\top \lap \xx{}$ from all sides of
% Eq.~\eqref{eq:lap-multiplicative-assumption} yields
% \[
%     \bigl(e^{-\varepsilon}-1\bigr)\,\xx{}^\top \lap \xx{}
%     \le
%     \xx{}^\top (\lap_\varepsilon-\lap) \xx{}
%     \le
%     \bigl(e^{\varepsilon}-1\bigr)\,\xx{}^\top \lap \xx{}.
% \]
Since $\lap\succeq 0$, we have $\xx{}^\top \lap \xx{}\ge 0$. Taking absolute values,
\[
    \bigl|\xx{}^\top (\lap_\varepsilon-\lap) \xx{}\bigr|
    \le
    \max\{\,e^{\varepsilon}-1,\;1-e^{-\varepsilon}\,\}\,
    \xx{}^\top \lap \xx{}.
\]
For $\varepsilon\ge 0$, note that
   $ e^{\varepsilon}-1 \;\ge\; 1-e^{-\varepsilon},$
since $e^{\varepsilon}+e^{-\varepsilon}\ge 2$. Therefore, $\bigl|\xx{}^\top (\lap_\varepsilon-\lap) \xx{}\bigr|
    \le
    \bigl(e^{\varepsilon}-1\bigr)\,\xx{}^\top \lap \xx{}.$
% \begin{equation*}
% \label{eq:lap-diff-quadratic-form}
%     \bigl|\xx{}^\top (\lap_\varepsilon-\lap) \xx{}\bigr|
%     \le
%     \bigl(e^{\varepsilon}-1\bigr)\,\xx{}^\top \lap \xx{}.
% \end{equation*}
% Taking absolute values, we obtain
% \begin{equation*}
% \label{eq:lap-diff-quadratic-form}
%     \bigl|\xx{}^\top (\lap_\varepsilon-\lap) \xx{}\bigr|
%     \le
%     \bigl(e^{\varepsilon}-1\bigr)\,\xx{}^\top \lap \xx{}.
% \end{equation*}
For any unit vector $\|\xx{}\|_2=1$, we further have
\[
    \bigl|\xx{}^\top (\lap_\varepsilon-\lap) \xx{}\bigr|
    \le
    \bigl(e^{\varepsilon}-1\bigr)\,\lambda_1(\lap)\,\|\xx{}\|_2^2
    =
    \bigl(e^{\varepsilon}-1\bigr)\,\|\lap\|_2.
\]
Taking the supremum over all $\|\xx{}\|_2=1$ and using the variational
characterization of the spectral norm for symmetric matrices, it follows that
\begin{equation*}
\label{eq:lap-diff-norm-bound}
    \|\lap-\lap_\varepsilon\|_2
    =
    \sup_{\|\xx{}\|_2=1}
    \bigl|\xx{}^\top (\lap_\varepsilon-\lap) \xx{}\bigr|
    \le
    \bigl(e^{\varepsilon}-1\bigr)\,\|\lap\|_2.\qedhere
\end{equation*}
\end{proof}

\subsection{Proof of Lemma \ref{lem:biadj-matrix}}
\label{proof of lem:biadj-matrix}
\begin{proof}
By the construction method described in Sec.\ref{sec:sparse-graph-gradient-code-encoding-matrix-construction}, the adjacency matrix of $G_\varepsilon$ is 
\begin{align*}
    \adj_\varepsilon =\begin{bmatrix} 0 & \GEP_\varepsilon \\[2pt] {\GEP_\varepsilon}^\top & 0 \end{bmatrix}\in\mathbb{R}^{2N\times 2N}.
\end{align*}

In Step 2, every row and every column of the pre–sparsified matrix (i.e. $\GEP$) sums to $d$, and Step~3 preserves weighted degrees at every vertex. So, we have $ \GEP_\varepsilon\ones{N} = d\,\ones{N}$,
% \text{ and } {\GEP_\varepsilon}^\top \ones{N} = d\ones{N},$
% % \begin{align*}
% %     \GEP_\varepsilon\ones{N} = d\,\ones{N} \qquad\text{and}\qquad {\GEP_\varepsilon}^\top \ones{N} = d\ones{N},
% % \end{align*}
% Therefore,
% \begin{align*}
%     \adj_\varepsilon\binom{\ones{N}}{\ones{N}}
% =\binom{\GEP_\varepsilon\ones{N}}{{\GEP_\varepsilon}^\top\ones{N}}
% =\binom{d\ones{N}}{d\ones{N}}
% =d\binom{\ones{N}}{\ones{N}}.
% \end{align*}
which implies $(\ones{N},\ones{N})$ is an eigenvector of $\adj_\varepsilon$ with eigenvalue $d$. It follows that $(\ones{N},\ones{N},d)$ is a singular triplet of $\GEP_\varepsilon$, since
% \begin{align*}
$    {\GEP_\varepsilon}^\top \GEP_\varepsilon\,\ones{N} = {\GEP_\varepsilon}^\top(d\ones{N}) = d\,{\GEP_\varepsilon}^\top \ones{N} = d^2 \ones{N},$
% \end{align*}
i.e., $\ones{N}$ is a right singular vector of $\GEP_\varepsilon$ with singular value $d$. This proves item (i).

For (ii), let $(\u_i(\adj_\varepsilon),\v_i(\adj_\varepsilon), \lambda_i(\adj_\varepsilon))$ be any eigenpair of $\adj_\varepsilon$ with $s,t\in\mathbb{R}^N$. Therefore,
\begin{align*}
   \adj
\begin{bmatrix}
\u_i(\adj_\varepsilon) \\
\v_i(\adj_\varepsilon)
\end{bmatrix}
=
\lambda
\begin{bmatrix}
\u_i(\adj_\varepsilon) \\
\v_i(\adj_\varepsilon)
\end{bmatrix}
\Longleftrightarrow
\begin{cases}
\GEP_\varepsilon\v_i(\adj_\varepsilon)=\lambda \u_i(\adj_\varepsilon),\\
{\GEP_\varepsilon}^\top \u_i(\adj_\varepsilon)=\lambda \v_i(\adj_\varepsilon).
\end{cases}
\end{align*}
Thus, we have
$
{\GEP_\varepsilon}^\top \GEP_\varepsilon \v_i(\adj_\varepsilon)=\lambda_i(\adj_\varepsilon)^2 \v_i(\adj_\varepsilon)
$ and $
\GEP_\varepsilon{\GEP_\varepsilon}^\top \u_i(\adj_\varepsilon)=\lambda_i(\adj_\varepsilon)^2 \u_i(\adj_\varepsilon)
$,
so the eigenvalues of $\adj$ are $\{\pm\sigma_i(\GEP_\varepsilon)\}_{i=1}^N$. Since the underlying bipartite graph is connected (the dense construction in Step~2 yields a connected graph and the degree-preserving sparsifier keeps it connected), the Perron–Frobenius theorem gives $\lambda_1(\adj_\varepsilon)=d$ and $\lambda_2(\adj_\varepsilon)<d$. Hence, $\sigma_1(\GEP_\varepsilon)=d$ and $\sigma_2(\GEP_\varepsilon)<d,$
% \begin{align*}
%     \sigma_1(\GEP_\varepsilon)=d
% \qquad\text{and}\qquad
% \sigma_2(\GEP_\varepsilon)<d,
% \end{align*}
which is item (ii).
\end{proof}

\subsection{Proof of Lemma \ref{lem:l2norm_of_aF}}
\label{proof of lem:l2norm_of_aF}
\begin{proof}
By construction,
\[
\ones{N}^\top \a_\F
= (N-S)\left(-\tfrac{S}{N-S}\right)+S\cdot 1
= -S+S
=0.
\]
Hence $\a_\F\perp\ones{N}$. By Lemma \ref{lem:biadj-matrix}, we have $\a_\F\perp\v_1(\GEP)$. Because
$\{\v_1(\GEP),\dots,\v_N(\GEP)\}$ is an orthonormal basis, this implies
\[
\a_\F\in \mathrm{span}\{\v_2(\GEP),\dots,\v_N(\GEP)\},
\]
so there exist $\alpha_2,\dots,\alpha_N$ with
$\a_\F=\sum_{i=2}^N \alpha_i\v_i(\GEP)$.

From the definition of $\a_\F$ (see Eq.\eqref{eq:decoding_vec}), we have
\[
\norm{\a_\F}_2^2
= (N-S)\left(\tfrac{S}{N-S}\right)^{\!2}
+ S\cdot 1^2
= \tfrac{S^2}{N-S}+S
= \tfrac{NS}{N-S}.
\]

Since $\{\v_i(\GEP)\}_{i=1}^N$ is orthonormal and $\a_\F\perp\v_1(\GEP)$, we have
$\norm{\a_\F}_2^2
= \textstyle\sum_{i=2}^N \bigl|\alpha_i\bigr|^2.$
Therefore, it follows that
$\sum_{i=2}^N \alpha_i^2=\tfrac{NS}{N-S}$. Thus,
$\norm{\a_\F}_2=\sqrt{\tfrac{NS}{N-S}}$.
\end{proof}

% \subsection{Proof of Lemma \ref{lem:l2norm_of_aF}}
% \label{proof of lem:l2norm_of_aF}
% \begin{proof}
% By construction,
% % \[
% $\ones{N}^\top \a_\F
% = (N-S)\left(-\tfrac{S}{N-S}\right)+S\cdot 1
% = -S+S
% =0.$
% % \]
% Hence $\a_\F\perp\ones{N}$. By Lemma \ref{lem:biadj-matrix}, we have $\a_\F\perp\v_1(\GEP)$. Because
% $\{\v_1(\GEP),\dots,\v_N(\GEP)\}$ is an orthonormal basis, this implies
% % \[
% $\a_\F\in \mathrm{span}\{\v_2(\GEP),\dots,\v_N(\GEP)\},$
% % \]
% so there exist $\alpha_2,\dots,\alpha_N$ with
% $\a_\F=\sum_{i=2}^N \alpha_i\v_i(\GEP)$.

% From the definition of $\a_\F$ (see Eq.\eqref{eq:decoding_vec}), we have
% \[
% \norm{\a_\F}_2^2
% = (N-S)\left(\tfrac{S}{N-S}\right)^{\!2}
% + S\cdot 1^2
% = \tfrac{S^2}{N-S}+S
% = \tfrac{NS}{N-S}.
% \]

% Since $\{\v_i(\GEP)\}_{i=1}^N$ is orthonormal and $\a_\F\perp\v_1(\GEP)$, we have
% $\norm{\a_\F}_2^2
% = \textstyle\sum_{i=2}^N \bigl|\alpha_i\bigr|^2.$
% Therefore, it follows that
% $\sum_{i=2}^N \alpha_i^2=\tfrac{NS}{N-S}$. Thus,
% $\norm{\a_\F}_2=\sqrt{\tfrac{NS}{N-S}}$.
% \end{proof}

\subsection{Proof of Lemma \ref{lemma:expectation-of-exponential-of-the-linear-of-a-Gaussian-random-variable}}
\label{proof of lemma:expectation-of-exponential-of-the-linear-of-a-Gaussian-random-variable}
\begin{proof} Let $G\sim \gaussian{\mu}{\sigma^2}$. The corresponding PDF is $\frac{1}{\sigma\sqrt{2\pi}}\int_{-\infty}^{\infty} \exp{(-\frac{1}{2}(\frac{x-\mu}{\sigma})^2)}$.
Therefore,
\begin{align*}
    &\ex{\exp{(\alpha G^2 + \theta G)}} \\
    &\quad= \tfrac{1}{\sigma \sqrt{2\pi}} \int_{-\infty}^{\infty} \exp{(\alpha x^2 + \theta x)}\exp{(-\tfrac{1}{2} (\tfrac{x-\mu}{\sigma})^2)} dx\\
    &\quad= \tfrac{1}{\sigma\sqrt{2\pi}}\int_{-\infty}^{\infty}\exp{\left(\alpha x^2 + \theta x - \tfrac{1}{2} (\tfrac{x-\mu}{\sigma})^2\right)} dx.
\end{align*}
Noting the identity
\begin{align*}
    &\alpha x^2 + \theta x - \tfrac{1}{2} (\tfrac{x-\mu}{\sigma})^2 = \alpha x^2 + \theta x - \tfrac{1}{2\sigma^2}(x^2 -2\mu x + \mu^2)\\
    % &\quad = (\alpha - \tfrac{1}{2\sigma^2}) x^2 + (\theta + 2\mu) x - \tfrac{\mu^2}{2\sigma^2}\\
    &\quad= \tfrac{2\alpha \sigma^2 - 1}{2\sigma^2} (x + \tfrac{\sigma^2(\theta + 2\mu)^2}{2\alpha \sigma^2 -1})^2- \tfrac{\sigma^2(\theta+2\mu)^2}{4\alpha \sigma^2-2} - 
    \tfrac{\mu^2}{2\sigma^2}
\end{align*}
and applying a change of variable, we obtain
\begin{align*}
    &\ex{\exp{(\alpha G^2 + \theta G)}} \\
    &= \exp{\big(- \tfrac{\sigma^2(\theta+2\mu)^2}{4\alpha \sigma^2-2} - \tfrac{\mu^2}{2\sigma^2}\big)} \tfrac{1}{\sigma\sqrt{2\pi}}\cdot \int_{-\infty}^{\infty}\exp{\big(-\tfrac{1/\sigma^2 - 2\alpha}{2} x^2\big)} dx\\
    &= \exp{\big(- \tfrac{\sigma^2(\theta+2\mu)^2}{4\alpha \sigma^2-2} - \tfrac{\mu^2}{2\sigma^2}\big)}\tfrac{1}{\sqrt{1/\sigma^2 - 2\alpha}}.
\end{align*}
If $G$ is a standard normal random variable, i.e., $G\sim \gaussian{0}{1}$, 
\begin{equation*}
    \ex{\exp{(\alpha G^2 + \theta G)}} = \tfrac{1}{\sqrt{1-2\alpha}} \exp{(\tfrac{\theta^2}{2- 4\alpha})}.\qedhere
\end{equation*}
\end{proof}

% \vspace{-1.5em}

\section{Algorithms for Degree Preserving Sparsifier}
\label{apx:short-cycle-decomposition}
We cite \textsc{DegreePreservingSparsify} and related Algorithms from \cite{degree-preserving-sparsifier} in this appendix.
It is worth noting that our construction permits positive real-valued edge weights. For practical implementation, we approximate the positive real edge weights using a fixed-point representation. Specifically, for a chosen precision parameter $\kappa = 2^{k}$, where $k\in \mathbb{N}$, we define
\[
\tilde{w}_e = \left\lfloor \kappa w_e \right\rceil \in \mathbb{N},
\]
where $\lfloor \cdot \rceil$ denotes rounding to the nearest integer. We then apply \textsc{DegreePreservingSparsify} to the integer weights $\tilde{w}_e$, and rescale the output by $1/\kappa$. 

This procedure introduces a quantization error of at most $1/\kappa$ per edge, which can be made arbitrarily small by choosing $k$ sufficiently large. In practice, this provides a tractable way to obtain integer weights while preserving the structure of the problem up to a controlled approximation.

% \begin{algorithm}[hp]
%     \caption{\textsc{DegreePreservingSparsify}$(G,\varepsilon)$}
%     \label{alg:degree-preserving-sparsification}
%     \KwIn{Graph $G$ with polynomially bounded edge weights.}
%     Choose a scaling factor $\kappa = 2^k$ (for some integer $k$) such that
%     $\kappa w_e \in \mathbb{N}$ for every edge weight $w_e$ of $G$.
    
%     Construct $G^{(\mathrm{int})}$ by scaling all edge weights:
%     for each edge $e$, set $w_e^{(\mathrm{int})} \gets \kappa w_e$.
    
%     Decompose each edge of $G$ by its binary representation. 
%     Now the edge weights of $G$ are powers of $2$ and are at most $m \log n$ in number.
    
%     Compute $r$, a $1.5$-approximate estimate of effective resistances in $G$.
    
%     \While{$|E(G)| \ge \Omega(\hat m \log n + nL\varepsilon^{-2}\log n)$}{
%        $G \gets \textsc{SparsifyOnce}(G, r, \textsc{ShortCycleDecomp})$
%     }
%     Obtain $G_\varepsilon$ by rescaling all edge weights of $G^{(\mathrm{int})}$ back:
%     for each edge $e$, set $w_e(G_\varepsilon) \gets w_e^{(\mathrm{int})}/\kappa$.
%     \Return $G$
% \end{algorithm}

\begin{algorithm}[!htbp]
\caption{\textsc{DegreePreservingSparsify}$(G,\varepsilon)$}
\label{alg:degree-preserving-sparsification}
\KwIn{Weighted graph $G$}
Choose $\kappa = 2^k$ such that $\kappa w_e \in \mathbb{N}$ for all edges $e$.

Scale weights: $w_e \gets \kappa w_e$ for all $e$.

Decompose edges so that all weights are powers of $2$.

Compute $r$, a $1.5$-approximation of effective resistances in $G$.

\While{$|E(G)| \ge \Omega(\hat m \log n + nL\varepsilon^{-2}\log n)$}{
    $G \gets \textsc{SparsifyOnce}(G, r, \textsc{ShortCycleDecomp})$
}
Rescale weights: $w_e \gets w_e / \kappa$ for all $e$.
\Return $G$
\end{algorithm}

\begin{algorithm}[!htbp]
\caption{\textsc{NaiveCycleDecomp}$(G)$}
\label{alg:naive-cycle-decomp}
\KwIn{Graph $G(V,E)$}
\KwOut{Cycle decomposition $\mathcal{C}$}

\While{$\exists\, u \in V$ with $\deg(u) \leq 2$}{
    Remove $u$ and its incident edges from $G$
}
\Return $\mathcal{C}$
\end{algorithm}

\begin{algorithm}[!htbp]
\caption{\textsc{ExtractBdDegGraph}$(G,\Delta)$}
\label{alg:extract-bounded-degree}
\KwIn{Graph $G(V,E)$, threshold $\Delta$}
\KwOut{Graph $H$ with degree $O(\Delta)$}

Initialize $H \gets \emptyset$.

\ForEach{vertex $u \in V$}{
   Add up to $\Delta$ incident edges of $u$ to $H$.
}

\ForEach{vertex $v$ in $H$ with $\deg(v) > 2\Delta$}{
   Split $v$ into $\left\lfloor \tfrac{\deg(v)}{\Delta} \right\rfloor$ vertices with balanced degrees.
}

\Return $H$
\end{algorithm}

% \begin{algorithm}[hp]
% \caption{\textsc{ShortCycleDecomp}$(G,l,k)$}
% \label{alg:short-cycle-decomp}
% \KwIn{Graph $G(V,E)$ with $n$ vertices and $m$ edges, number of recursion layers $l$, reduction factor $k$.}
% \KwOut{A decomposition of $E$ into a union of cycles $\mathcal{C}$ and a set $E_{\text{extra}}$ of extra edges.}

% \If{$l = 0$ \textbf{ or } $|V(G)| < k$}{
%    \Return \textsc{NaiveCycleDecomp}$(G)$\;
% }

% Set $\Delta \gets \left(64\cdot10^6\gamma_{\text{NS}}(n)^4 \log^2(2n)\right)^l k$, where $\gamma_{\text{NS}}(n) = \exp(O(\sqrt{\log n\log\log n}))$ for graphs with $2n$ vertices.

% Initialize $\mathcal{C}$ and $E_{\text{extra}}$ as empty.

% \While{$G$ has vertices remaining}{
%    (a) Repeatedly remove any vertex of $G$ with degree $< \Delta$, and add its incident edges to $E_{\text{extra}}$.
%    (b) $H \gets \textsc{ExtractBdDegGraph}(G,\Delta)$.
%    (c) $(S,\mathcal{C}_{\text{partial}}) \gets \textsc{MoveEdges}(H,k)$.
%    (d) Create $H_S$ from $S$ and $\mathcal{C}_{\text{partial}}$ (as described in Lemma 8.2).
%    (e) $\mathcal{C}_{H_S} \gets \textsc{ShortCycleDecomp}(H_S,l-1,k)$.
%    (f) Extend each cycle in $\mathcal{C}_{H_S}$ to a circuit in $H$ and in turn $G$ via $\mathcal{C}_{\text{partial}}$ (according to Lemma 8.2). Split these circuits into cycles, add them to $\mathcal{C}$, and remove them from $G$.
% }
% \Return{$\mathcal{C},E_{\text{extra}}$}
% \end{algorithm}

\begin{algorithm}[!htbp]
\caption{\textsc{ShortCycleDecomp}$(G,l,k)$}
\label{alg:short-cycle-decomp}
\KwIn{Graph $G(V,E)$, recursion depth $l$, parameter $k$}
\KwOut{Cycles $\mathcal{C}$ and extra edges $E_{\text{extra}}$}

\If{$l = 0$ \textbf{or} $|V| < k$}{
   \Return \textsc{NaiveCycleDecomp}$(G)$
}

Set $\Delta \gets \left(64\cdot10^6\,\gamma_{\text{NS}}(n)^4 \log^2 n\right)^l k$.

Initialize $\mathcal{C} \gets \emptyset$, $E_{\text{extra}} \gets \emptyset$.

\While{$G$ is non-empty}{
   Remove vertices of degree $< \Delta$ and add their incident edges to $E_{\text{extra}}$.

   $H \gets \textsc{ExtractBdDegGraph}(G,\Delta)$

   $(S,\mathcal{C}_{\text{partial}}) \gets \textsc{MoveEdges}(H,k)$

   Construct $H_S$ from $(S,\mathcal{C}_{\text{partial}})$

   $\mathcal{C}_{H_S} \gets \textsc{ShortCycleDecomp}(H_S,l-1,k)$

   Lift $\mathcal{C}_{H_S}$ to cycles in $G$, add to $\mathcal{C}$, remove from $G$
}

\Return{$(\mathcal{C}, E_{\text{extra}})$}
\end{algorithm}

\begin{algorithm}[!htbp]
\caption{\textsc{MoveEdges}$(G,k)$}
\label{alg:move-edges}
\KwIn{Graph $G=(V,E)$, parameter $k$}
\KwOut{Vertex set $S$ and a partial cycle decomposition onto $S$}

Set $\alpha \gets d_{\min}/(4\gamma_{\text{NS}}(n))$.

$(E^s,E^d) \gets \textsc{NSExpanderDecompose}(G,\alpha)$.

Initialize $S \gets \emptyset$.

\ForEach{connected component $H$ of $E^s$}{
   \eIf{$|V(H)| \le k$}{
      Apply \textsc{NaiveCycleDecomp}$(H)$ and record the resulting cycles.
   }{
      Apply \textsc{MoveEdgesExpander}$(H,\alpha/d_{\max},k)$ and record the output.
   }
}

\Return $S$ and the union of the recorded partial decompositions
\end{algorithm}

% \begin{algorithm}
% \caption{\textsc{MoveEdgesExpander}$(G, \phi, k)$}
% \label{alg:move-edges-expander}
% \KwIn{Undirected unweighted graph $G=(V,E)$ with $n \geq k$ vertices and $m \geq n$ edges; 
%       $\phi$, a lower bound on the conductance of $G$; 
%       reduction factor $k$.}
% \KwOut{A set $S$ of $\lceil n/k \rceil$ vertices and a partial cycle decomposition of $G$ onto $S$.}

% Pair up multiple edges to form cycles.

% Pick $S$ consisting of the $\lceil n/k \rceil$ vertices of maximum degree in $G$.

% \ForEach{edge $e=(u,v)\in E$}{
%    (a) Generate $4k$ lazy random walks each from $u$ and $v$, of length $10\phi^{-2}\log n$. 
%    A lazy random walk stays at the current vertex with probability $1/2$ at each step.
%    (b) Discard walks that use $e$.
%    (c) Pick any one walk each from $u$ and $v$ (if there is one) that terminates in $S$ and convert them into simple paths. 
%    Add the corresponding cycle in $G/S$ to consideration
% }

% Greedily pick a set of edges whose corresponding cycles are edge-disjoint.

% \If{fewer than $\tfrac{\phi^4 m}{2\cdot 10^3 k \log^2 n}$ cycles are formed}{
%    Go to Step 3
% }
% \Return $S$ and the partial cycle decomposition
% \end{algorithm}

\begin{algorithm}[!htbp]
\caption{\textsc{MoveEdgesExpander}$(G,\phi,k)$}
\label{alg:move-edges-expander}
\KwIn{Graph $G=(V,E)$, conductance lower bound $\phi$, parameter $k$}
\KwOut{Vertex set $S$ and a partial cycle decomposition onto $S$}

Pair parallel edges to form cycles.

Set $S$ as the $\lceil n/k \rceil$ highest-degree vertices.

\ForEach{edge $e=(u,v) \in E$}{
   Generate $4k$ lazy random walks\footnote{A lazy random walk stays at the current vertex with probability $1/2$ at each step.} from each of $u$ and $v$ of length $O(\phi^{-2}\log n)$.

   Discard walks traversing $e$.

   If both endpoints yield walks terminating in $S$, form a cycle via the induced paths in $G/S$.
}

Select a maximal set of edge-disjoint cycles.

\If{fewer than $\Omega\!\left(\frac{\phi^4 m}{k\log^2 n}\right)$ cycles are obtained}{
   repeat the process
}

\Return $S$ and the selected cycles
\end{algorithm}

\FloatBarrier
% \section*{Acknowledgment}
% \jossie{TODO}

% \newpage
\bibliographystyle{IEEEtran}
\bibliography{bibliofile.bib}

% \begin{thebibliography}{1}
% \bibliographystyle{IEEEtran}

% \bibitem{ref1}
% {\it{Mathematics Into Type}}. American Mathematical Society. [Online]. Available: https://www.ams.org/arc/styleguide/mit-2.pdf

% \bibitem{ref2}
% T. W. Chaundy, P. R. Barrett and C. Batey, {\it{The Printing of Mathematics}}. London, U.K., Oxford Univ. Press, 1954.

% \bibitem{ref3}
% F. Mittelbach and M. Goossens, {\it{The \LaTeX Companion}}, 2nd ed. Boston, MA, USA: Pearson, 2004.

% \bibitem{ref4}
% G. Gr\"atzer, {\it{More Math Into LaTeX}}, New York, NY, USA: Springer, 2007.

% \bibitem{ref5}M. Letourneau and J. W. Sharp, {\it{AMS-StyleGuide-online.pdf,}} American Mathematical Society, Providence, RI, USA, [Online]. Available: http://www.ams.org/arc/styleguide/index.html

% \bibitem{ref6}
% H. Sira-Ramirez, ``On the sliding mode control of nonlinear systems,'' \textit{Syst. Control Lett.}, vol. 19, pp. 303--312, 1992.

% \bibitem{ref7}
% A. Levant, ``Exact differentiation of signals with unbounded higher derivatives,''  in \textit{Proc. 45th IEEE Conf. Decis.
% Control}, San Diego, CA, USA, 2006, pp. 5585--5590. DOI: 10.1109/CDC.2006.377165.

% \bibitem{ref8}
% M. Fliess, C. Join, and H. Sira-Ramirez, ``Non-linear estimation is easy,'' \textit{Int. J. Model., Ident. Control}, vol. 4, no. 1, pp. 12--27, 2008.

% \bibitem{ref9}
% R. Ortega, A. Astolfi, G. Bastin, and H. Rodriguez, ``Stabilization of food-chain systems using a port-controlled Hamiltonian description,'' in \textit{Proc. Amer. Control Conf.}, Chicago, IL, USA,
% 2000, pp. 2245--2249.

% \end{thebibliography}

\newpage

% \section{Biography Section}
% If you have an EPS/PDF photo (graphicx package needed), extra braces are
%  needed around the contents of the optional argument to biography to prevent
%  the LaTeX parser from getting confused when it sees the complicated
%  $\backslash${\tt{includegraphics}} command within an optional argument. (You can create
%  your own custom macro containing the $\backslash${\tt{includegraphics}} command to make things
%  simpler here.)
 
% \vspace{11pt}

% \bf{If you include a photo:}\vspace{-33pt}
% \begin{IEEEbiography}[{\includegraphics[width=1in,height=1.25in,clip,keepaspectratio]{fig1}}]{Michael Shell}
% Use $\backslash${\tt{begin\{IEEEbiography\}}} and then for the 1st argument use $\backslash${\tt{includegraphics}} to declare and link the author photo.
% Use the author name as the 3rd argument followed by the biography text.
% \end{IEEEbiography}

% \vspace{11pt}

% \bf{If you will not include a photo:}\vspace{-33pt}
% \begin{IEEEbiographynophoto}{John Doe}
% Use $\backslash${\tt{begin\{IEEEbiographynophoto\}}} and the author name as the argument followed by the biography text.
% \end{IEEEbiographynophoto}

\vfill

\end{document}